\documentclass[letterpaper]{article}
\usepackage[top=1in, bottom=0.9in, left=0.9in, right=1in]{geometry}
\usepackage{parskip} 
\usepackage{amssymb}
\usepackage{makecell}

\usepackage{fullpage}
\usepackage[dvips]{color}
\usepackage{xcolor,float}
\usepackage{latexsym}
\usepackage[activeacute,english]{babel}
\usepackage{graphicx}
\usepackage{amsmath}
\usepackage{bm}
\usepackage{nccmath}
\usepackage{natbib}
\usepackage{amsthm,appendix,url,inconsolata}
\usepackage{algorithmic}
\usepackage{algorithm}
\usepackage{mathtools}
\usepackage{enumitem,wrapfig}
\usepackage{multirow,sidecap}

\usepackage{amscd}
\usepackage{amsfonts}
\usepackage{amssymb}
\usepackage[tableposition=top]{caption}
\usepackage{wrapfig}
\usepackage{ifthen}
\usepackage[utf8]{inputenc}
\usepackage{ulem} 
\usepackage{xspace} 

\usepackage{afterpage}
\usepackage{setspace}
\usepackage[customcolors]{hf-tikz}
\usepackage{authblk} 
\usepackage{hyperref}

\usepackage{caption}
\usepackage{subcaption}
\def\E{\mathbb{E}}
\newtheoremstyle{mytheoremstyle}
  {\topsep} 
  {\topsep} 
  {} 
  {} 
  {\bfseries} 
  {.} 
  {.5em} 
  {} 
\theoremstyle{mytheoremstyle}
\newtheorem{theorem}{Theorem}
\newtheorem{definition}[theorem]{Definition}
\newtheorem{proposition}[theorem]{Proposition}
\newtheorem{corollary}{Corollary}[theorem]

\allowdisplaybreaks[4]

\DeclareMathOperator*{\argmin}{\arg\!\min}

\graphicspath{{./art/}}

\usepackage{booktabs}
\usepackage{colortbl}

\colorlet{tableheadcolor}{gray!25} 

\newcolumntype{K}[1]{>{\centering\arraybackslash}m{#1}}

\usepackage{natbib}

\title{Statistical summaries of unlabelled evolutionary trees and ranked hierarchical clustering trees}

\author[1]{Rajanala Samyak}
\author[1,2,*]{Julia A. Palacios}
\affil[1]{Department of Statistics, Stanford University, Stanford, CA 94305, USA}
\affil[2]{Department of Biomedical Data Science, Stanford Medicine, Stanford, CA 94305, USA}
\affil[*]{Corresponding author email: juliapr@stanford.edu}
\date{\today}

\begin{document}
\onehalfspacing
\maketitle

\begin{abstract}
	Rooted and ranked binary trees are mathematical objects of great importance used to model hierarchical data and evolutionary relationships with applications in many fields including evolutionary biology and genetic epidemiology.  
	Bayesian phylogenetic inference usually explore the posterior distribution of trees via Markov Chain Monte Carlo methods, however assessing uncertainty and summarizing distributions or samples of such trees remains challenging.
	While labelled phylogenetic trees have been extensively studied, relatively less literature exists for unlabelled trees which are increasingly useful, for example when 
	one seeks to summarize samples of trees obtained with different methods, or from different samples and environments, and wishes to assess stability and generalizability of these summaries.  
	In our paper, we exploit recently proposed distance metrics of unlabelled ranked binary trees and unlabelled ranked genealogies (equipped with branch lengths) to define the Fr\'{e}chet mean and variance as summaries of these tree distributions. 
	We provide an efficient combinatorial optimization algorithm for computing the Fr\'{e}chet mean from a sample of or distribution on unlabelled ranked tree shapes and unlabelled ranked genealogies. 
	We show the applicability of our summary statistics for studying popular tree distributions and for comparing the SARS-CoV-2 evolutionary trees across different locations during the COVID-19 epidemic in 2020. 
\end{abstract}



\section{Introduction}

Binary trees are used to represent the ancestral relationships of samples in evolving populations and to model other type of dependent data that is the result of tree-like structures, such as transmission trees. 
In phylogenetics, population genetics, and cell biology, evolutionary trees are inferred from observed molecular sequence data or from observed evolutionary traits in a sample of individuals from a population. 
These individuals can be viral sequences from infected hosts (viral phylodynamics), species (phylogenetics), individuals from a single species (population genetics), or cells (cancer evolution). 
The estimated tree describes the inferred ancestral relationships of the samples and provides information about the past evolutionary dynamics of the sample's population. 
For example, in the context of viral phylodynamics, the tree provides information about the past transmission history and pathogenesis \citep{volz2013viral}.

\begin{figure}[H]
	\centering
	\includegraphics[width=\textwidth]{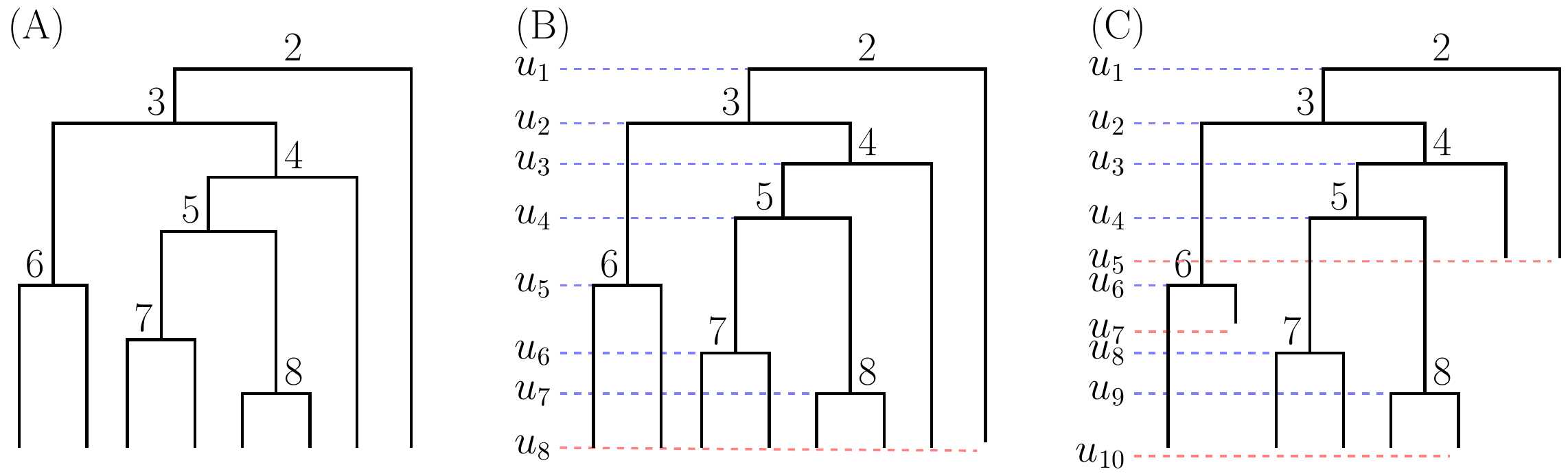}
	\caption{\textbf{Tree examples.} (A) isochronous ranked tree shape (topology only); (B) isochronous ranked genealogies (ranked tree shape and branching times); and (C) heterochronous ranked genealogy with different sampling times (red dashed lines).}
	\label{fig:rut_example}
\end{figure}

Distance-based summaries of labelled and unranked tree structures have been extensively studied such as phylogenetic trees \citep{hillis2005analysis, chakerian2012computational, benner2013computing, Willis2018, Owen2019}, and hierarchical clustering trees such as CART \citep{breiman1984algorithm} in the last few decades.  
These summaries rely on metrics such as the BHV distance \citep{Billera2001} defined on tree space. \citet{KuhnerYam} presents a comparison of different distances on such space. 
In the case of labelled trees, it is usually assumed that all sampled trees have the same set of leaves (labels). 
Other type of summaries of labelled and unranked trees include bootstrap confidence levels \citep{felsenstein1985confidence, efron1996bootstrap}, and maximum clade credibility trees \citep{heled2013looking, o2018efficacy}.

In this article, we are interested in summarizing different tree structures, namely those that are ranked and unlabelled.
These trees are useful in the study of the ancestral relationships of a sample of objects which are exchangeable.
One potential area of application is in the study of cancer evolution where the tree represents the evolutionary history of many cells in a patient's tumor.
We may want to summarize many such trees, each tree inferred from a different patient, in order to find a ``representative" tree and quantify how much variation or heterogeneity is present across patients with the same type of cancer or across different types of cancer.  
Similar type of questions have been studied assuming a coarser type of tree structure \citep{govek2018consensus}.  
Another application is in infectious disease transmission, where we study the mutation history of different viral sequences.
Here, we compute summary statistics from samples of SARS-CoV-2 trees using RNA sequences available in the data repository GISAID between February and September 2020.

Recent mathematical results concerning distance metrics on the space of ranked evolutionary trees without leaf labels enable quantitative comparisons of evolutionary trees of different sets of organisms living at different geographic locations and different time periods.
These metrics have been used for comparing posterior distributions (and testing for distribution equality) of trees of human influenza A virus from two different geographic regions; for summarizing  empirical distributions via medoids; and for defining credible sets of posterior distributions \citep{Kim2019}. 

In this article, we use the previously defined distance metrics on unlabelled ranked evolutionary trees to understand distributional properties of some popular tree models used in Biology and to summarize samples of trees. 
Tree samples can be either obtained from the posterior distribution for a given sample of molecular sequences such as those obtained with \texttt{BEAST} \citep{Suchard2018}, or trees independently obtained in different studies. 

A recently proposed Bayesian method for inferring evolutionary parameters from molecular data is based on the Tajima coalescent for unlabelled ranked evolutionary trees \citep{Palacios2019, cappello2020tajima}. 
This method approximates the posterior distribution of model parameters and evolutionary trees by Markov chain Monte Carlo methods (MCMC). 
While summarizing real-valued parameters is straightforward, summarizing a posterior sample of unlabelled ranked evolutionary trees is more challenging. 
As discussed above, methods for summarizing posterior samples of labelled evolutionary trees exist \citep{heled2013looking, benner2013computing, Owen2019}.  
However, we are not aware of methods applicable for summarizing posterior samples of unlabelled ranked trees.

To summarize samples and populations of unlabelled ranked evolutionary trees, we define the sample and population Fr\'{e}chet means and variances in terms of the recently proposed $d_{1}$ and $d_{2}$ distances. 
We compare these summaries to others measures of centrality and dispersion on the same space. 
We show that finding the Fr\'{e}chet mean consists in finding the solution of an integer programming problem. 
However, the search space of ranked evolutionary trees grows super-exponentially in $n$, the number of leaves, which makes the problem of finding the Fr\'{e}chet mean computationally challenging for large values of $n$. 
For large $n$, we rely on stochastic combinatorial optimization, and propose a simulated annealing algorithm for estimating the Fr\'{e}chet mean in the case of both isochronous and heterochronous ranked tree shapes.  
The time complexity for computation of these trees is independent of $N$, the number of trees, and instead only depends on $n$, the number of leaves in each tree. 
The advantage of our approach is that the output is not restricted to the sample, and can be any element of the space. 

We also introduce exploratory tools for the space of unlabelled ranked trees such as a total ordering, a Markov chain on the space, and credible balls.  These are used in various analyses in this paper but also can also merit separate study.

In Section \ref{sec:prelim} we formally define the class of evolutionary trees considered, show how these trees are in bijection with a particular set of triangular matrices of integers, and show how to compute the distance metrics proposed in \citep{Kim2019}. 
Section \ref{sec:mean} introduces the Fr\'{e}chet mean as a summary of location and methods for computing and approximating it, including a simulated annealing algorithm using a Markov chain for state space exploration. 
In Section \ref{sec:var} we describe the Fr\'{e}chet variance as a statistic of spread for a given sample or distribution. 
In Section \ref{sec:results} we compare various summary statistics of theoretical distributions as well as empirical simulated distributions. 
As a real data example, we analyse the posterior distributions of evolutionary trees inferred from SARS-CoV-2 molecular sequences from the states of California, Texas, Florida, and Washington.  We obtain Fr\'echet mean trees for different samples and display multidimensional scaling plots to visualize intra-state and inter-state variability.
In Section \ref{sec:discussion}, we conclude and discuss future directions. 

We have developed an R package, \textsc{fmatrix} (\verb|github.com/RSamyak/fmatrix|), which implements the various methods discussed in this paper.  
The package is compatible with \textsc{phylodyn} \citep{phylodynRpackage}, an R package for phylodynamic simulation and inference, and \textsc{ape} \citep{paradis2019ape}, an R package to handle phylogenetic trees.


\section{Preliminaries} \label{sec:prelim}

\textit{Ranked unlabelled trees} or \textit{ranked tree shapes} are rooted binary trees, with an increasing ordering of the interior nodes.
They are \textit{unlabelled} in the sense that the external nodes (leaves) are unlabelled.  
We however rank the internal nodes, starting at the root with label $2$ (Figure \ref{fig:rut_example} (A)). 
A ranked unlabelled tree, additionally equipped with the vector of branching times is called a \textit{ranked genealogy} (Figure \ref{fig:rut_example} (B)).

A ranked unlabelled tree (or ranked genealogy) is called \textit{isochronous} if all the leaves are sampled at the same time, usually assumed to be sampled at time $0$ (Figure \ref{fig:rut_example} (A-B)). 
In applications of rapidly evolving pathogens such as Influenza A virus, molecular sequences (leaves) are sampled at different times (Figure \ref{fig:rut_example} (C)) and these trees are called \textit{heterochronous}.


Recently, \citet{Kim2019} proposed metrics on the space of ranked tree shapes and ranked genealogies for comparing and assessing differences between tree distributions. 
The utility of the metrics was demonstrated in a series of simulation studies to assess differences between the posterior distributions of ranked genealogies of Influenza A virus obtained from two different geographic regions. 
The proposed metrics rely on a unique representation of ranked tree shapes as triangular matrices of integers, called $\mathbf{F}$-matrices. 
The distances between two ranked unlabelled trees are then calculated as $L_1$ and $L_2$ distances between these matrices (Eqs. \ref{eq:dist1} and \ref{defn:d2-metric}).  
A formal definition of $\mathbf{F}$-matrices is given in Theorem~\ref{prop:bijection}.

For a given isochronous ranked tree shape with $n$ leaves, we denote the time of a branching event at each node $i$ by $u_{i-1}$ and the time interval between the two subsequent nodes $i$ and $i+1$ by $I_i=(u_{i-1}, u_{i})$. For convenience, we assign $u_n=0$ at the leaves. An $\mathbf{F}$-matrix representation of a ranked tree shape with $n$ leaves is an $(n-1) \times (n-1)$ triangular matrix of non-negative integers. The diagonal elements of the $\mathbf{F}$-matrix indicate the number of branches at each time interval. The off-diagonal element $F_{i,j}$, $2 \leq j < i \leq n-1$, represents the number of branches extant at $I_j=(u_{j+1}, u_{j})$ that do not bifurcate during the interval $(u_{i+1}, u_{j})$. Figure \ref{fig:fmat5} shows all ranked tree shapes with 5 leaves (first row) and their corresponding $\mathbf{F}$-matrix encodings (second row).

\begin{figure}[H]
	\centering
	\includegraphics[width=6in]{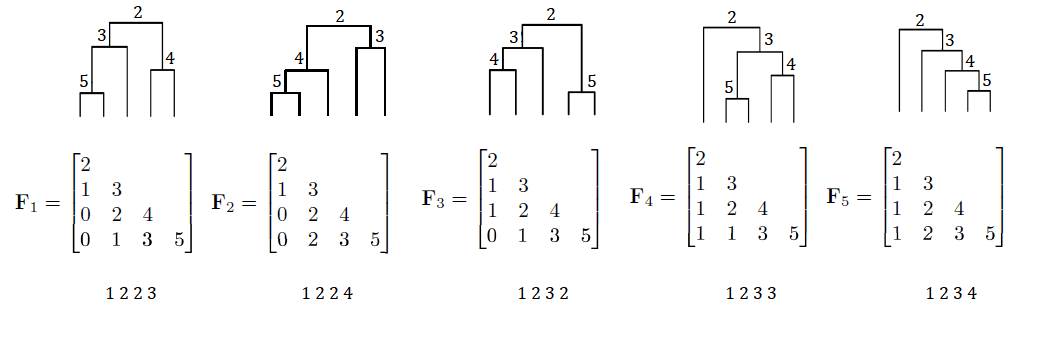}
	\caption{\textbf{All ranked tree shapes with $n = 5$ leaves.} The second row shows the corresponding $\mathbf{F}$-matrix representation of the ranked tree shape of the first row, and the third row shows their corresponding functional code representations (see section \ref{app:mc-encod}
		of the appendix).}
	\label{fig:fmat5} 
\end{figure}

\begin{theorem}
	\label{prop:bijection}
	The space of ranked tree shapes with $n$ leaves $\mathcal{T}_n$ is in bijection with the space $\mathcal{F}_n$ of $(n-1)\times(n-1)$  $\mathbf{F}$-matrices, which are lower triangular square matrices of non-negative integers that obey the following constraints:
	
	1. The diagonal elements are $F_{i,i} = i+1$ for $i = 1, \dots, n-1$ and the subdiagonal elements are $F_{i+1, i} = i$ for $i = 1, \dots, n-2$.
	
	2. The elements $F_{i,1}, i=3, \dots, n-1$, in the first column satisfy $\max \{0, F_{i-1, 1} - 1\} \leq F_{i,1} \leq F_{i-1, 1}$.
	
	3.  All the other elements $F_{i,k}, i= 4,\dots, n-1$ and $k = 2, \dots, i-2$ satisfy the following inequalities: 
	\begin{align*}
		\max\{0, F_{i, k - 1}\} &\leq F_{i,k} \\
		F_{i-1, k} - 1 &\leq F_{i,k} \leq F_{i-1, k}  \\
		F_{i, k - 1} + F_{i-1, k} - F_{i-1,k-1} - 1 &\leq F_{i,k} \leq F_{i,k-1} + F_{i-1, k} -F_{i-1,k-1}
	\end{align*}
\end{theorem}

The proof is in the appendix, in Section \ref{app:proof_bijection}.


Given two ranked tree shapes $y_1$ and $y_2 \in \mathcal{T}_{n}$, with corresponding $\mathbf{F}$-matrices $F_1$ and $F_2$, $d_{1}$ and $d_{2}$ are the $L_1$ and the $L_{2}$ distances on the space of matrices, restricted to $\mathcal{F}_n$, the class of $\mathbf{F}$-matrices, that is

\begin{equation} \label{eq:dist1}
	d_1(y_1, y_2) := d_1(F_1, F_2) = \sum_{i,j}{|(F_1)_{ij} - (F_2)_{ij}|}
\end{equation}
and 
\begin{equation} 
	d_2(y_1, y_2) := d_2(F_1, F_2) = \sqrt{\sum_{i,j}{\left((F_1)_{ij} - (F_2)_{ij}\right)^2}}
	\label{defn:d2-metric}
\end{equation}

For ranked genealogies $G_1$ and $G_2 \in \mathcal{G}_{n}$, with  corresponding $F_1$ and $F_2$, $\mathbf{F}$-matrices 
$$ d_1(G_1, G_2) := \sum_{i,j}{|(F_1)_{ij}(W_1)_{ij} - (F_2)_{ij}(W_2)_{ij}|},$$ and
$$ d_2(G_1, G_2) := \sqrt{\sum_{i,j}{((F_1)_{ij}(W_1)_{ij} - (F_2)_{ij}(W_2)_{ij})^{2}}},$$ where $W_1$ and $W_2$ are weight matrices constructed using the respective branching event times $u_k = (u_{k,n}, u_{k, n-1}, \dots, u_{k,1})$ of the $k$-th tree, $k = 1,2$ with $(W_k)_{ij} := |u_{k, j} - u_{k, i+1}|$ (Figure \ref{fig:rut_example}(B)).

For defining a distance on the space of heterochronous ranked tree shapes or genealogies (Figure \ref{fig:rut_example}(C)), \citet{Kim2019} propose supplementing the $\mathbf{F}$-matrix with additional rows for sampling events. 
The heterochronous $d_{1}$ and $d_{2}$ distances are then computed analogously to the isochronous distances, as the $L_{1}$ and $L_{2}$ distances between the extended $\mathbf{F}$-matrices of the same size, as detailed in Section 3 of the appendix of \citet{Kim2019}.  
For improving computational efficiency when computing pairwise distances among a large number of trees, we modify the distance slightly and consider all trees together when adding additional rows to the $\mathbf{F}$-matrix.  
Further details can be found in Section \ref{app:additional_events} 
in the appendix.


\section{Central summaries} \label{sec:mean}

Let $y_{1},\ldots,y_{m} \in \mathcal{T}_{n}$ be $m$ ranked tree shapes with $n$ leaves independently drawn from a common probability distribution. 
We are interested in summarizing such samples and identifying a representative ranked tree shape of the sample. 
Similarly, given a probability distribution over the space of ranked tree shapes, we are interested in knowing what is the ``expected'' tree of that distribution. 

In evolutionary biology applications, trees are usually not directly observed. 
Instead, researchers either find a single tree via maximum likelihood estimation or maximum parsimony \citep{felsenstein2004inferring}, or generate a large sample of trees from the posterior distribution via Markov chain Monte Carlo methods \citep{Suchard2018}. 
In the latter case of Bayesian inference, it is not clear how to best summarize a sample of ranked unlabelled trees and define a representative ``mean" tree as a measure of centrality. 
A decision theoretic approach is to use the $d_{1}$ and the $d_{2}$ distance metrics to define the absolute and the squared error loss functions. 
Minimizing the expected posterior losses gives us the posterior median and posterior mean respectively. 
However, the value that minimizes the expected posterior loss in Euclidean space does not usually correspond to a tree.  In particular for  $d_2$, the Euclidean mean of the F-matrices may not even be a matrix of integers. 
However, we can restrict the minimizer to be an element of the space.
The resulting summary is the sample posterior Fr\'{e}chet mean. Similarly, in applications when a random sample of trees or a population of trees is available, the Fr\'{e}chet mean would be the tree in the space that minimizes the empirical mean loss and the expected loss respectively.

Current practices for summarizing labelled trees include reporting a majority-rule consensus tree (MRC), a maximum clade credibility (MCC), and a median tree based on metrics on labelled trees \citep{benner2013computing,Owen2019}. 
The MRC is obtained by choosing partitions with probability greater than $0.5$ from the list of observed partitions and it is usually annotated with marginal probabilities of each partition as a measure of uncertainty \citep{cranston2007summarizing}. 
The MCC tree is the tree with the maximum product of clade probabilities and it is arguably, the most used central summary of labelled trees. 

The concept of consensus partition is not longer applicable for ranked tree shapes. 
Instead, we rely on the proposed distances on ranked tree shapes to define the Fr\'{e}chet mean in the same line as the median tree for labelled trees. 
The Fr\'{e}chet mean is the tree in the space that has the minimum expected squared distance to a tree in the space, and hence it provides a natural notion of central tree. 
We extend the notion to ranked genealogies, including heterochronous genealogies.

\subsection{The Fr\'{e}chet mean}
\label{ss:mean}

We first consider the metric spaces $(\mathcal{T}_{n},d)$, where $d$ is either $d_{1}$ or $d_{2}$, and let $\mu$ denote a finite probability mass function on $\mathcal{T}_{n}$, then the barycenter of $\mu$, also called a Fr\'{e}chet mean tree \citep{frechet1948elements}, is any element $\bar{T}\in \mathcal{T}_{n}$ such that 
\begin{equation} \label{defn-t-mean}
	\bar{T} \in \argmin_{x \in \mathcal{T}_{n}} \sum_{y \in \mathcal{T}_{n}} d(x,y)^{2}\mu(y).
\end{equation}

Since $\mathcal{T}_n$ is finite, we immediately have the existence of the minimizer in Equation \ref{defn-t-mean}.  Note that uniqueness may not be guaranteed, as though the objective function is convex, the space is discrete.  If there is more than one minimizer, the elements of the Fr\'echet mean set will be close to each other in the metric $d$.  However, for many generating tree distributions, we observe that the Fr\'echet mean is unique and for the rest of the paper we consider as our summary one element of the Fr\'echet mean set.

Examples of probability distributions on $\mathcal{T}_{n}$ include the Tajima coalescent (Yule model) \citep{Tajima1983, veber} in which $\mu(T)=\mu(F)=\frac{2^{n-c-1}}{(n-1)!}$, where $c$ is the number of cherries, i.e., the number of internal nodes subtending two leaves. We consider a larger class of probability distributions on ranked tree shapes in Section~\ref{sec:results}.

Similarly, for the metric spaces $(\mathcal{G}_{n},d)$ with $d$ corresponding to $d_{1}$ or $d_{2}$ on $\mathcal{G}_{n}$ and $\nu$ a probability measure on $\mathcal{G}_{n}$ such that 
\begin{equation*}
	\int_{\mathcal{G}_{n}}d(x,y)^{2}d\nu(y) < \infty,
\end{equation*}
the Fr\'{e}chet mean genealogy is any element $\bar{G} \in \mathcal{G}_{n}$ such that 
\begin{equation} \label{defn-g-mean}
	\bar{G} \in \argmin_{G \in \mathcal{G}_{n}} \int_{H \in \mathcal{G}_{n}} d(G,H)^{2}d\nu(H).
\end{equation}

In this manuscript, we will consider probability models on isochronous genealogies such that for $G=(F,\mathbf{u})$,  $d\nu(G)=\mu(F)\prod^{n-1}_{j=1}f(u_{j} \mid u_{j+1})d(\mathbf{u})$, that is, the tree topology and the branching event times are independent. In evolutionary biology applications, this assumption corresponds to neutral evolution in a closed population \citep{wakeley_coalescent_2008}. In this case, the Fr\'{e}chet mean becomes:

\begin{equation} \label{eq:frech_gen}
	\bar{G} \in \argmin_{G \in \mathcal{G}_{n}} \sum_{F^{H} \in \mathcal{F}_{n}} \int_{0}^{\infty} \int_{u_{n-1}}^{\infty} \cdots \int^{\infty}_{u_{2}} d(G,H)^{2}\mu(F^{H})f(u_{n-1}\mid u_{n})\cdots f(u_{1}\mid u_{2})du_{1}du_{2}\cdots du_{n-1}.
\end{equation}

A remarkable property of the mean (Eq.~\ref{eq:frech_gen}) under the $d_{2}$ distance on genealogies, is that the optimization problem can be separated into two optimization problems, one for finding the tree topology and one for finding the branching event times. The following proposition formalizes this result. 

\begin{proposition}
	Let $\nu$ be a probability measure on $\mathcal{G}_n$, the space of isochronous genealogies, such that the tree topology and branching event times are independent under $\nu$.  The Fr\'{e}chet mean $\bar{G}_2 = (F^*, u^*)$ under the $d_2$ metric can be obtained by separately finding $F^*$ and $u^*$.
	\label{prop:separate_top_time}
\end{proposition}

The proof is in the appendix, in Section \ref{app:proof_separate_top_time}.

The empirical Fr\'{e}chet mean of a given sample $y_1, \dots y_m$ from the metric space $(\mathcal{T}_{n},d)$ is obtained by taking $\mu$ in (\ref{defn-t-mean}) to be the empirical measure.  The mean of a sample $h_1, \dots, h_m$ from $(\mathcal{G}_{n},d)$ is obtained analogously, where $\nu$ in (\ref{defn-g-mean}) is taken to be the empirical measure, that is

\begin{align} \label{defn-mean-sample}
	\bar{T} \in \argmin_{x \in \mathcal{T}_{n}} \sum_{j = 1}^m d(x,y_j)^{2} \\
	\bar{G} \in \argmin_{g \in \mathcal{G}_{n}} \sum_{j = 1}^m d(g,h_j)^{2}
\end{align}

The cardinality of the space $\mathcal{T}_n$ is given by the Euler zigzag numbers (OEIS A000111), which grow super-exponentially with $n$, ${\displaystyle |\mathcal{T}|_{n}\sim 2\left({{2}/{\pi }}\right)^{n+1}\cdot n!}$.  Finding the Fr\'{e}chet mean is hence  computationally challenging for large $n$.  

For a simpler summary of centrality, we can use an in-sample version of (\ref{defn-mean-sample}), which we call the \textit{restricted Fr\'{e}chet mean}:

\begin{align} \label{defn-medoid-sample}
	\bar{T}^{\text{in-sample}}_{i} \in \argmin_{x \in \{y_1, \dots y_m\}} \sum_{j = 1}^m d_{i}(x,y_j)^{2} \\
	\bar{G}^{\text{in-sample}}_{i} \in \argmin_{g \in \{h_1, \dots h_m\}} \sum_{j = 1}^m d_{i}(g,h_j)^{2}
\end{align}

This may be reasonable for spaces with large number of leaves when direct computation of the Fr\'{e}chet mean is not be possible. However, constraining ourselves to stay only within the sample may be undesirable.

\subsection{Mixed Integer Programming} \label{ss:mip}

In principle, the definition of the Fr\'{e}chet mean is not sensitive to the choice of metric. However, in the case of the $d_{1}$ and $d_{2}$ metrics, the Fr\'{e}chet mean is the minimizer of a convex objective function. In addition, the space of $\mathbf{F}$-matrices is characterized by a set of linear inequalities (Theorem \ref{prop:bijection}) and hence the problem of finding the Fr\'{e}chet mean can be framed as a mixed integer programming problem. 

Let $\mu$ be a probability measure on $\mathcal{T}_n$, or equivalently on $\mathcal{F}_n$, then 
in the case of the $d_2$ metric, the Fr\'{e}chet mean $\bar{F}_2$ is given by:

\begin{align*}
	\bar{F}_2 &\in 
	\underset{F \in \mathcal{F}_n}
	{\operatorname{argmin}} 
	\sum_{H \in \mathcal{F}_n}
	{\sum_{k,l} (F_{kl} - H_{kl})^2} \mu(H) \\
	&= \underset{F \in \mathcal{F}_n}{\operatorname{argmin}} 
	\sum_{H \in \mathcal{F}_n}
	{\sum_{k,l} {(F_{kl}^2 - 2F_{kl}H_{kl} )} \mu(H) }\\
	&= \underset{F \in \mathcal{F}_n}{\operatorname{argmin}} 
	{\sum_{k,l} \{F_{kl}^2 - 2F_{kl}M_{kl}} \}
\end{align*}

\begin{figure}[H]
	\centering
	\includegraphics[width=0.3\textwidth]{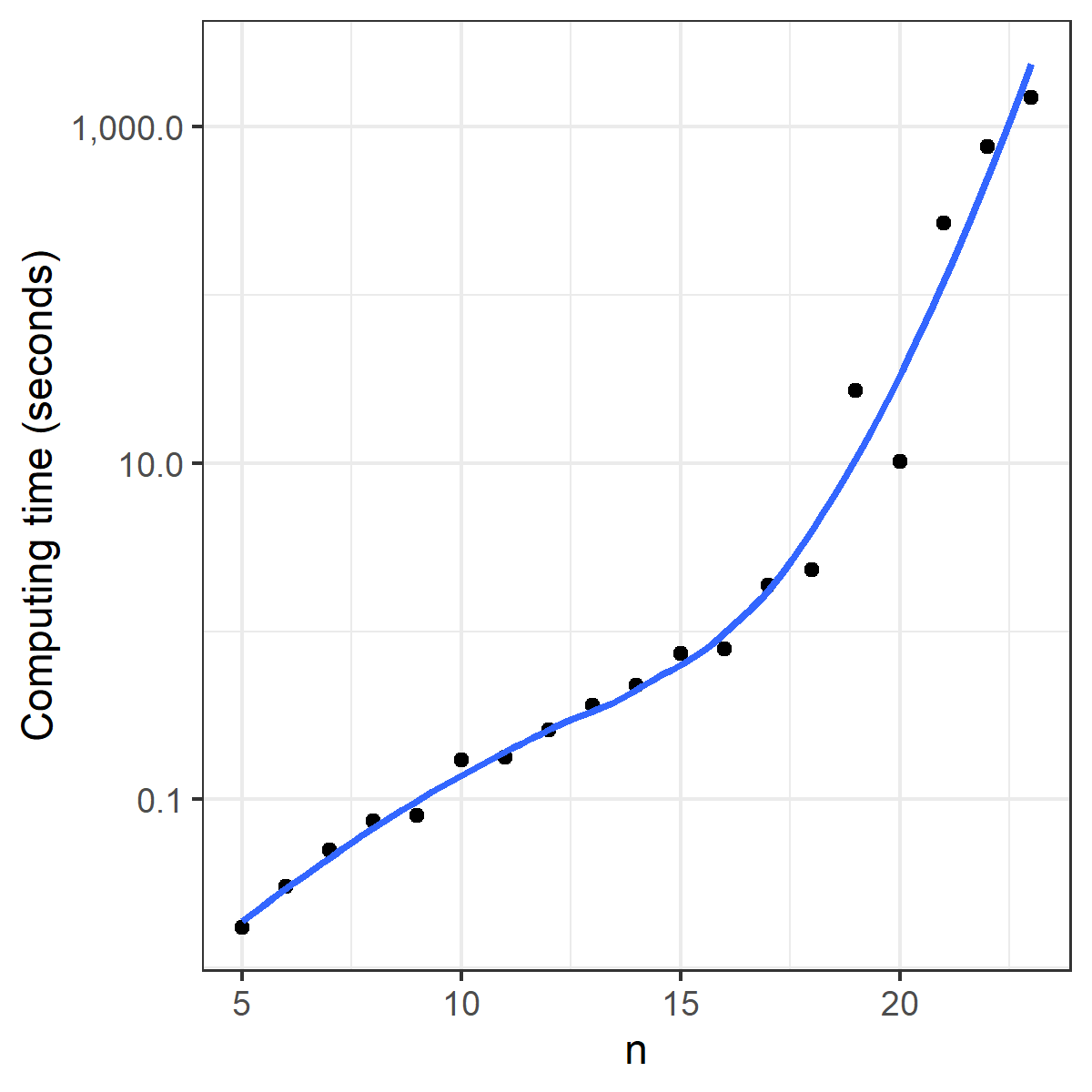}
	\caption{\textbf{Running time.}  Exact computation for Fr\'{e}chet mean under the Yule model using Gurobi, plotted against dimension of $\mathbf{F}$-matrices.  Computations done on laptop with an Intel i7 processor.}
	\label{fig:gurobi_time}
\end{figure} 
where 
$M_{kl} = \sum_{H \in \mathcal{F}_n}{H_{kl}\cdot \mu(H)}.$ This is a simple quadratic objective function with linear and integrality constraints.  We  use \textsc{gurobi} \citep{gurobi}, a standard MIP solver,  to directly perform the optimization. The implementaion of the code is available in the \texttt{R} package \textsc{fmatrix}. Note that once the means $M_{kl}$ are computed, the rest of the problem no longer involves the $m$ samples. That is, the problem scales with the number of leaves $n$ but not with the number of samples $m$. This is particularly important when summarizing  samples obtained through MCMC, since $m$ is usually of high order.

This method works well for small number of leaves $n$, such as $n = 20$, but it quickly becomes impractical with larger $n$. Figure \ref{fig:gurobi_time} shows how the computation time grows exponentially in $n$. For larger $n$, we resort to stochastic combinatorial optimization algorithms that scale well at the expense of solution guarantees as discussed in Section \ref{sec:sa}.

\subsection{Simulated annealing algorithm} \label{sec:sa}

When the number of leaves is large, the MIP solution is computationally demanding and often unfeasible.  
A simple technique that works well is simulated annealing \citep{kirkpatrick1983optimization}, which is a  general-purpose stochastic algorithm for optimizing an objective function over a potentially large discrete set. 
Simulated annealing explores the ranked tree shape space via a Metropolis-Hastings algorithm.
We trade the guarantee of an exact solution for computational tractability.

In order to describe the simulated annealing algorithm, we define two Markov chains on the space of ranked tree shapes (one for isochronous and one for heterochronous trees) in Section \ref{app:mc-encod} 
of the appendix.  The two Markov chains are then used as proposal distributions in the Metropolis-Hastings step of the simulated annealing algorithm.

In the case of the Fr\'{e}chet mean, SA aims to minimize the energy function $E(x) = \sum_{i=1}^m {d(x, y_i)^2}$ for a sample of trees $\{y_{i}\}^{m}_{i=1}$ over $x\in \mathcal{T}_{n}$. This problem is equivalent to finding the maximum of $\exp\{-E(x)/R\}$ at any given temperature $R>0$. Let $\{R_{k}\}$ be a sequence of monotone decreasing temperatures such that  $\lim_{k \rightarrow \infty} R_{k}=0$, for example $R_k = \alpha^k R_0$ for some high initial temperature $R_0$, and $\alpha < 1$.  (This is the exponential cooling schedule.) Then, at each temperature, the SA algorithm consists of Metropolis-Hastings (MH) steps that targets $\pi_{k}(x) \propto \exp\{-E(x)/R_{k}\}$ as the stationary distribution. As the number of steps increases, $\pi_{k}(x)$ puts more and more of its probability mass in the set of global maxima.  SA differs from descent algorithms by allowing transitions to higher energy states at higher temperatures, in order to avoid being stuck at local maxima.

In our implementations (Algorithm~\ref{alg:sa}) for isochronous and heterochronous Fr\'{e}chet means, the MH proposal distributions are given by the transition kernels of the Markov chains described in section \ref{app:mc-encod} 
of the appendix.  In both cases, the transition kernels are symmetric and hence, the MH acceptance probability of moving from $x_{k-1}$ to $x_{k}$ is given by:

$$a_{k}= \exp{\left(-\dfrac{E(x_{k})}{R_{k}} + \dfrac{E(x_{k-1})}{R_{k-1}} \right)} \wedge 1. $$

The temperature schedule in SA needs to be specified and affects the time taken for convergence of the algorithm.  Theoretical convergence guarantees exist for the logarithmic cooling schedule $R_k = R_0(1 + \alpha \log(1 + k))^{-1}$
with sufficiently high initial temperature and appropriately chosen $\alpha$ (see Chapter 3 of \citet{aarts1988simulated}).  However this schedule is prohibitively slow for most problems.  In practice, we observe that the exponential cooling schedule with $\alpha$ chosen very close to $1$ performs reasonably well. The benefits of simulated annealing are its easy implementation and the design of the algorithm which allows getting out of local optima.

\begin{algorithm}
	\caption{Fr\'{e}chet mean of a sample of ranked unlabelled trees via simulated annealing \label{alg:sa}}
	\label{alsf}
	\begin{algorithmic}
		\REQUIRE $T_1, \dots T_m$ sample of ranked unlabelled trees, starting position $T^{(0)}$, initial temperature $R_0 > 0$, decay parameter $\alpha \in (0,1)$.
		
		\STATE Define energy function $E(T) = \sum_{i=1}^m {d(T, T_i)^2}$.  \textit{($d$ is a metric defined in Section \ref{sec:prelim})} 
		
		\STATE $k \leftarrow 0$
		
		\REPEAT 
		\STATE $S \leftarrow \text{random neighbour of } T^{(k)}$  \textit{(generate proposal using Markov
			chains of Definitions \ref{def:mc_defn} or \ref{def:mc_defn_hetero}, for isochronous and heterochronous trees respectively)}
		\IF{$\texttt{runif}(1) < \exp\left(-\dfrac{E(S) - E(T^{(k)})}{R_k} \right)$} 
		\STATE $T^{(k+1)} \leftarrow S$  \textit{(accept)}
		\ELSE
		\STATE $T^{(k+1)} \leftarrow T^{(k)}$ \textit{(reject)}
		\ENDIF
		\STATE $R_{k+1} \leftarrow \alpha R_k$ \textit{(reduce temperature)}
		\STATE $k \leftarrow k+1$
		\UNTIL{convergence of $T^{(k)}$}
		
	\end{algorithmic}
\end{algorithm}

Using the result shown in Section \ref{ss:mip}, we can replace the energy function $E(T) = \sum_{i=1}^m {d(T, T_i)^2}$ in the case of the $d^2$ metric by $E(T) = \|F-M\|^2$ where $F$ is the $\mathbf{F}$-matrix corresponding to $T$ and $M$ is the Euclidean mean of the $\mathbf{F}$-matrices corresponding to the $T_i$.

We note that for both isochronous and heterochronous genealogies, our algorithm first finds the average coalescent times and then finds the tree topology via the SA algorithm just described. In the case of heterochronous genealogies, the Markov chain used is conditioned on a fix set of coalescent and sampling times. We analyse the computational performance of the SA algorithm in the appendix.


\section{Measures of dispersion} \label{sec:var}

In this section we define and discuss three notions for quantifying uncertainty or dispersion in a distribution or a sample of ranked tree shapes (or ranked genealogies).

The \textbf{Fr\'{e}chet variance} is a natural measure of dispersion for arbitrary probability metric spaces. It measures the concentration around the Fr\'{e}chet mean. While the Fr\'{e}chet mean is an object of the metric space, the Fr\'{e}chet variance is a simple scalar summary of spread. 

Let $\mu$ be a probability measure on $\mathcal{T}_n$ ($\mathcal{F}_n$).  The \textit{Fr\'{e}chet variance}  of $y \sim \mu$ with respect to the metric $d$ is defined as follows:

\begin{equation}
	V = \underset{y \in \mathcal{T}_n}{\sum}{d(y, \bar{T})^2 \cdot \mu(y)}, \quad \text{where} \quad \bar{T} = \underset{x \in \mathcal{T}_n}{\operatorname{argmin}} {\underset{y \in \mathcal{T}_n}{\sum}{d(x, y)^2 \cdot\mu(y)}}
\end{equation}

Let $y_1, \dots y_m \in \mathcal{T}_n$ be a random sample of ranked tree shapes.  Then the \textit{sample Fr\'{e}chet variance} is given by  

\begin{equation}
	V_{m} = \dfrac{1}{m}\sum_{i=1}^m {d(y_i, \bar{T})^2}, \quad \text{where} \quad \bar{T} = \underset{x \in \mathcal{T}_n}{\operatorname{argmin}} {\sum_{i=1}^m {d(x, y_i)^2}}
\end{equation}

Similarly, the Fr\'{e}chet variance of $G \sim d \nu$ can be obtained by integrating over the probability space of branching event times and ranked tree shapes.

Another scalar measure of dispersion is \textbf{entropy} \citep{mezardmontanari2009}.  Entropy is a function of the probability measure only and it does not depend on the metric $d$.  A measure with zero entropy is concentrated on a single point, and a large entropy indicates greater uncertainty in the position of a random variable with the underlying measure.

Let $\mu$ be a probability measure on $\mathcal{T}_n$.  The \textit{entropy} of the space is given by 

\begin{equation}
	H = - \underset{y \in \mathcal{T}_n}{\sum} {\mu(y)\cdot \log\left[\mu(y)\right]}
\end{equation}

The discrete distribution with maximum entropy is the uniform distribution. While entropy is meaningful at the population level, it does not provide a meaningful value for a given sample. In Figure \ref{fig:frechet_all} we compare entropy with Fr\'{e}chet variance for a particular class of probability models on ranked tree shapes.

In many applications, a single mean value and the variance are not enough for summarizing the distribution. \textbf{Interquartiles and credible intervals} are typically used to inform about the concentration of the distribution around the central value for real-valued distributions. The analogues in the space of ranked tree shapes are defined as follows: 

A \textit{central interquartile ball} of ranked tree shapes of level $1-\alpha$, $\alpha \in [0,1]$ is the set 
\begin{equation*}
	B_{\varepsilon}(\bar{T}) \overset{d}{=} 
	\{y \in \mathcal{T}_n : d(y, \bar{T}) \leq \varepsilon \}, 
\end{equation*}
where $\varepsilon \text{ is the smallest } \epsilon \geq 0 \text{ such that } \text{P}(B _{\epsilon}(\bar{T}))\geq 1-\alpha$, where $\bar{T}$ is a point estimate.

Similarly, a level $1-\alpha$ \textit{credible ball} is the set $B_{\varepsilon}(\bar{T})$ where $\varepsilon \text{ is the smallest } \epsilon \geq 0 \text{ such that } \text{P}(B _{\epsilon}(\bar{T})\mid \mathcal{D})\geq 1-\alpha$.

Although credible and interquartile balls can be defined in a meaningful way in terms of the $d_{1}$ and $d_{2}$ distances to the mean value, summarizing meaningful boundaries of the sets in this space is challenging. One attempt to meaningfully define boundaries for credible sets and interquartile sets is through a total ordering on $\mathcal{T}_n$.  Such an ordering roughly corresponds to a one-dimensional projection of the space, and the boundaries of a credible set can be taken to be the extreme points of the set with respect to the ordering.

We propose an ordering based on the distance to a reference ranked tree shape, for example the Fr\'{e}chet Mean $T_K$ of the Kingman model \citep{Kingman1982}, which is a commonly used neutral model for evolution, together with a lexicographic order in the $\mathbf{F}$-matrix representation.  We construct our ordering in such a way that the \textit{most unbalanced} tree $T_{\text{unb}}$ (also called caterpilar tree) and the \textit{most balanced} tree $T_{\text{bal}}$ are two poles of the order, and the Fr\'{e}chet mean $T_K$ lies somewhere in between those two poles.

To be more precise, the most unbalanced tree denoted here as $T_{\text{unb}} \in \mathcal{T}_{n}$ is the only ranked tree shape with one cherry, i.e., one internal node that subtends two leaves. The unbalanced tree with $n=5$ leaves is depicted in the last column of Figure \ref{fig:fmat5}.

\begin{proposition} The ranked tree shape at maximum $d_{1}$ and $d_{2}$ distances to the unbalanced tree $T_{\text{unb}}\in \mathcal{T}_{n}$ is $T_{\text{bal}}$
	with the following $\mathbf{F}$-matrix encoding:
	\begin{equation}\label{eq:fbal}
		F^{(\text{bal})}=\begin{bmatrix}
			2 & \\
			1 & 3 \\
			0 & 2 & 4 \\
			0 & 1 & 3 & 5 \\
			0 &  0 & 2 & 4 & 6 \\
			\vdots & \vdots & \vdots & \vdots &  \vdots &\ddots & \\
			0 & 0 & 0 & 0 & \cdots & n-2 & n\\
		\end{bmatrix},
	\end{equation}
	that is, $F_{i,j}^{(\text{bal})}=\max\{0,2j-i+1\}$ for $i=1,\ldots,n-1$ and $j=1,\ldots,i$, and  and $F^{(\text{bal})}_{i,j}=0$ for $i=1,\ldots, n-1$ and $j=i,\ldots,n-1$ (upper triangle).
\end{proposition}
\begin{proof}
	First note that the most unbalanced tree has the following $\mathbf{F}$-matrix encoding:
	\begin{equation}
		F^{(\text{unb})}=\begin{bmatrix}
			2 \\
			1 & 3 \\
			1 & 2 & 4 \\
			1 & 2 & 3 & 5 \\
			1 &  2 & 3 & 4 & 6\\
			\vdots & \vdots & \vdots & \vdots &  \vdots &\ddots & \\
			1 & 2 & 3 & 4 & \cdots & n-2 & n\\
		\end{bmatrix},
	\end{equation}
	that is, $F^{(\text{unb})}_{i,j}=j$ for $i=2,\ldots,n-1$, and  $j=1,\ldots,i-1$ and $F^{(\text{unb})}_{i,j}=0$ for $i=1,\ldots, n-1$ and $j=i+1,\ldots,n-1$ (upper triangle) and $F^{(\text{unb})}_{i,i}=i+1$ for $i=1,\ldots,n-1$ (diagonal).  Further, note that for any $F\in \mathcal{F}_{n}$, the values in each row are non-decreasing to the right, i.e. $F_{i,j} \leq F_{i,j+1}$ (constraint 3, Theorem \ref{prop:bijection})  and the values in each column are non-increasing, i.e. $F_{i,j} \geq F_{i+1,j}$ (constraints 2 and 3, Theorem \ref{prop:bijection}). Second, $F^{(\text{unb})}_{i,j}\geq F_{i,j}$ for all $F \in \mathcal{F}_{n}$ and $i,j \leq n-1$, that is, $F^{(\text{unb})}$ has the largest $d_{1}$ and $d_{2}$ norms. Third, note that $F^{(\text{bal})}_{i,j}\leq F_{i,j}$ for all $F \in \mathcal{F}_{n}$ and $i,j \leq n-1$, that is, $F^{(\text{bal})}$ has the smallest $d_{1}$ and $d_{2}$ norms. Moreover, $F^{(\text{bal})}_{i,j} \leq F^{(\text{unb})}_{i,j}$ for all $i,j \leq n-1$ and the pair: $T_{\text{unb}}$ and $T_{\text{bal}}$ have the largest $d_{1}$ and $d_{2}$ distances among all pairwise distances in $\mathcal{T}_{n}$.
\end{proof}

We note that the ranked tree shape corresponding to $F^{(\text{bal})}$ is called the most balanced ranked tree shape for ease of interpretation. However, there may be arguably many more similarly balanced trees in the population.

We will now define the \textit{signed-distance} as the distance to a reference ranked tree shape $\bar{T}$ with an additional sign depending on whether the tree is closer to the most unbalanced or to the most balanced tree.

\begin{definition}(The signed-distance function to $\bar{T}$). Let $f(x): \mathcal{T}_{n} \rightarrow \mathbb{R}^{+}$ and $\bar{T} \in \mathcal{T}_{n}$ a reference ranked tree shape, such that 
	\begin{equation}
		f(x)=\begin{cases}
			-d(x,\bar{T}) & \text{ if } d(x,T_{\text{unb}})\leq d(x,T_{\text{bal}})\\
			d(x,\bar{T}) & \text{ if } d(x,T_{\text{unb}})> d(x,T_{\text{bal}})\\
		\end{cases}
	\end{equation}
\end{definition}

The signed-distance induces a partial order on $\mathcal{T}_{n}$, however since many ranked tree shapes can have the same signed distance to $\bar{T}$, many pairs of trees will be incomparable . We will say that  $T_{1} \sim T_{2}$ belong to the same equivalence class if $f(T_{1})=f(T_{2})$. When a set of ranked tree shapes belong to the same equivalence class, we will order the ranked tree shapes in the equivalence class according to their lexicographic order using a vectorized $F$ representation as follows.

\begin{definition}(Lexicographic order).\label{def:lex_ord} Let $F^{(1)}=(F^{(1)}_{1,1},F^{(1)}_{2,1},\ldots,F^{(1)}_{1,n-1}, F^{(1)}_{2,2},F^{(1)}_{2,3},\ldots,F^{(1)}_{n-1,n-1})$ be the column-vectorized representation of $T_{1} \in \mathcal{T}_{n}$, and let $F^{(2)}=(F^{(2)}_{1,1},F^{(2)}_{2,1},\ldots,F^{(2)}_{1,n-1}, F^{(2)}_{2,2},F^{(2)}_{2,3},\ldots,F^{(2)}_{n-1,n-1})$ be the column-vectorized representation of $T_{2} \in \mathcal{T}_{n}$. We say that $T_{1} \preceq_{\text{lex}} T_{2}$ if $F^{(1)}=F^{(2)}$ or the first non-vanishing difference $F^{(1)}_{i}-F^{(2)}_{i}$ is positive for $i=1,\ldots,m$, $m=\frac{n(n-1)}{2}$. 
\end{definition}

For example $T^{(\text{unb})} \preceq_{\text{lex}} T^{(\text{bal})}$ and $T^{(\text{unb})} \preceq_{\text{lex}} T \preceq_{\text{lex}} T^{(\text{bal})}$ for any $T \in \mathcal{T}_{n}$. Although the lexicographic order is a total order, we propose to order all ranked tree shapes in the space according to their signed distance to $T_{K}$ and to only use the lexicographic order within equivalence classes as follows.

\begin{definition}\label{def:order} We say that $T_{1} \preceq T_{2}$ if $f(T_{1})< f(T_{2})$ or if $f(T_{1})=f(T_{2})$ and $T_{1} \preceq_{\text{lex}} T_{2}$.
\end{definition}

\begin{proposition} The order induced by $\preceq$ of Definition \ref{def:order} on $\mathcal{T}_{n}$ is a total order.
\end{proposition}
\begin{proof}
	To show antisymmetry note that the only way $T_{1} \preceq T_{2}$ and $T_{2} \preceq T_{1}$ is that $f(T_{1})=f(T_{2})$ and $T_{1} \preceq_{\text{lex}} T_{2}$ and $T_{1} \preceq_{\text{lex}} T_{2}$. This occurs only if $F^{(1)}=F^{(2)}$. The bijection of Theorem~\ref{prop:bijection} then implies that  $T_{1} = T_{2}$. Transitivity and convexity follow directly from the transitivity and convexity of $<$ and $\preceq_{\text{lex}}$.
\end{proof}

We note that $\preceq_{\text{lex}}$ is not the only possible lexicographic order; for example a row-vectorized representation of $T\in \mathcal{T}_{n}$ can be replaced in definition \ref{def:lex_ord} to generate another ordering. Although the lexicographic order is not a biologically meaningful order, it provides a consistent way for comparing histograms across different tree models on the same space (see for example, the third row of Figure \ref{fig:theoretical}).

Having established the $\preceq$ order, we will summarize credible balls and interquartile balls by at most four ranked tree shapes and by at least two ranked tree shapes. Let $B_{\varepsilon}(\bar{T})$ denote the interquartile or credible ball and $\bar{T}$ the Fr\'{e}chet mean of the distribution. Then the set of ranked tree shapes at the boundary of $B_{\varepsilon}(\bar{T})$ will be partitioned into two sets, one with positive signed distance to $\bar{T}$ and one with negative signed distance to $\bar{T}$. If the cardinality of the sets is greater than one, we will then summarize each set by the smallest and the largest ranked tree shape in each set according to the $\preceq$ order.


\section{Results} \label{sec:results}


\subsection{Statistical summaries of Blum-Fran\c{c}ois distributions on ranked tree shapes} \label{ss:results_dist}

We present and analyse point summaries of a large family of ranked tree shape models called the Blum-Francois $\beta$-splitting model.  After the introduction of Aldous' $\beta$-splitting model 
on cladograms (tree shapes without rankings) \citep{Aldous1996,Aldous2001}, many extensions and generalizations of this model have been proposed on different resolutions of trees, including the Blum-Fran\c {c}ois model on ranked tree shapes \citep{Sainudiin2016} and the alpha-beta splitting model \citep{Maliet2018}. 
Henceforth, we will refer to the Blum-Fran\c {c}ois model on ranked tree shapes simply as the \textbf{Blum-Fran\c {c}ois model}.

\begin{figure}[H]
	\centering
	\includegraphics[width=0.3\textwidth]{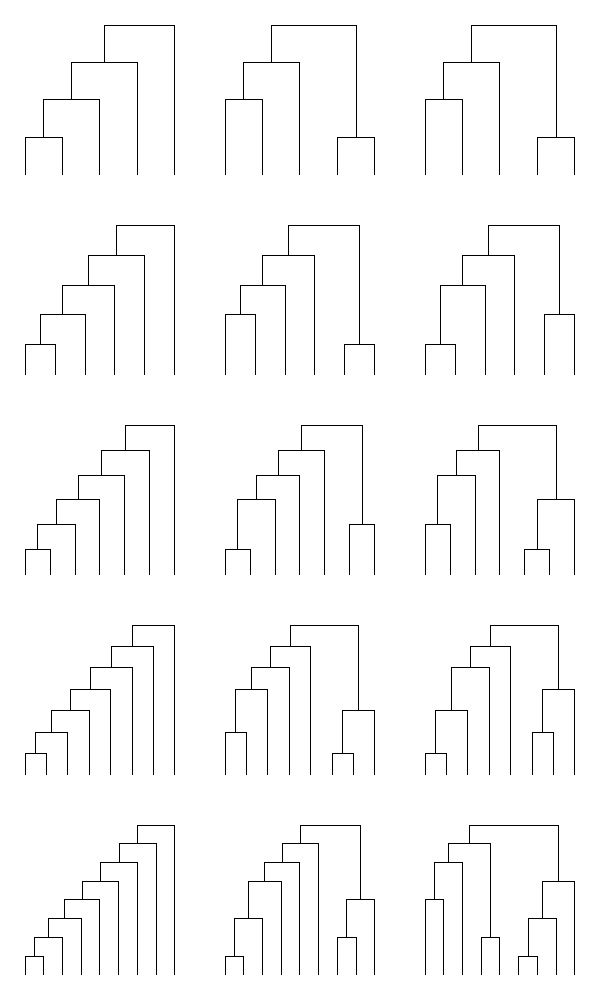}
	\caption{\textbf{Fr\'{e}chet mean under the Blum-Fran\c {c}ois $\beta$-splitting model.}  Top to bottom: $n = 5, \dots, 9$, Left to right: $\beta = -1, 0, 100.$}
	\label{fig:means_list_small_n}
\end{figure}
We start with the Blum-Fran\c {c}ois model on ranked unlabelled \textit{planar} trees (there is distinction between left and right offspring). Let $n_i^L$ and $n_i^R$ denote the number of internal nodes in the left and right subtrees below node $i$. In particular, if node $i$ is a cherry, i.e. subtends two leaves, then $n_i^L = n_i^R = 0$. Then, a ranked unlabelled planar tree with $n$ leaves has probability mass function given by:
$$P (T_{\text{planar}}) = \prod_{i=1}^{n-1} \dfrac{B(n_i^L + \beta + 1, n_i^R + \beta + 1)}{B(\beta+1, \beta+1)} $$ where $B(a,b) = \int_0^1 {x^{a-1} (1-x)^{b-1} dx}$ is the Beta function and $\beta \in [-1,\infty)$. A ranked tree shape $T$ obtained by ignoring the distinction between left and right subtrees has then the following probability mass function:
$$P (T) = 2^{n-1-c}\prod_{i=1}^{n-1} \dfrac{B(n_i^L + \beta + 1, n_i^R + \beta + 1)}{B(\beta+1, \beta+1)}$$ where $c$ is the number of cherries in $T$. The $\beta$ parameter controls the level of balancedness of the distribution. In particular, when $\beta=0$, the corresponding distribution $P(T)=2^{n-1-c}/(n-1)!$ is the coalescent distribution on ranked tree shapes (Kingman coalescent), also known as the Yule distribution.

\begin{figure}[H]
	\centering
	\includegraphics[width=\textwidth]{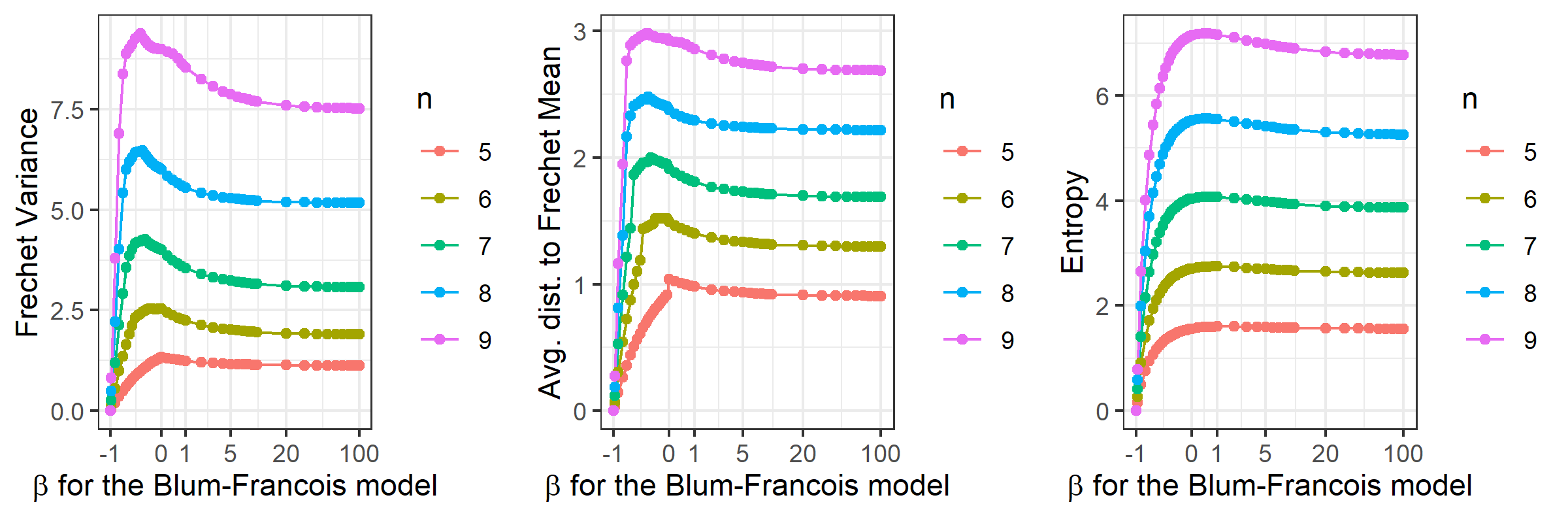}
	\caption{\textbf{Measures of dispersion.} Fr\'{e}chet variance, average distance to Fr\'{e}chet mean, and entropy, for small $n$ under the Blum-Fran\c{c}ois $\beta$-splitting model
	}
	\label{fig:frechet_all}
\end{figure}

Figure \ref{fig:means_list_small_n} shows the $d_{2}$ Fr\'{e}chet means of the ranked tree shapes distributions with $5,\ldots,9$ leaves under the Blum-Fran\c{c}ois distribution with $\beta \in \{-1,0,100\}$. For small values of $\beta$, the distribution generates unbalanced trees and for large values of $\beta$, the distribution generates balanced trees. Figure \ref{fig:frechet_all} shows the $d_{2}$ Fr\'{e}chet variance, expected distance to the Fr\'{e}chet mean and entropy for $n = 5, \dots, 9$ and $\beta$ values spaced out in $[-1, \infty]$. For $\beta>1$, the variance and entropy remain relatively constant as functions of $\beta$. The largest variance and entropy are obtained when $\beta\leq 0$. We explicitly compute the corresponding  Fr\'{e}chet mean, Fr\'{e}chet variance, and entropy for small $n$ by enumerating all the ranked tree shapes and evaluating their probability mass functions. 

\begin{figure}[H]
	\centering
	\includegraphics[height=0.8\textheight]{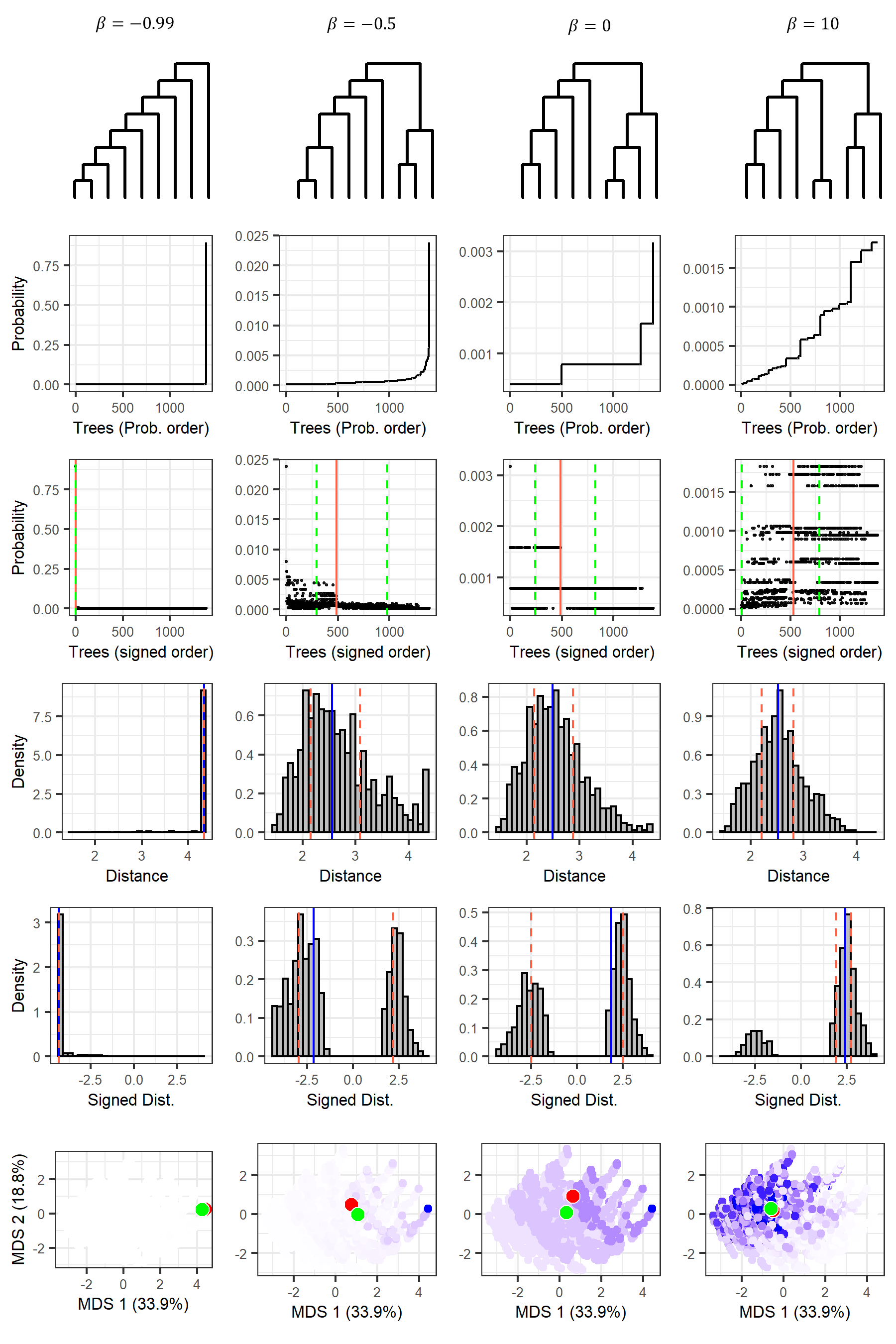}
	\caption{\textbf{Summarizing Blum-Fran\c{c}ois distributions on ranked tree shapes.} Blum-Fran\c{c}ois distributions on ranked tree shapes with $n = 9$ leaves, columns correspond to $\beta = -0.99, -0.5, 0~ \text{(coalescent)}, 10$ respectively. \emph{Row 1:} Fr\'{e}chet Mean of the distribution; \emph{Row 2:} Probability mass function of trees, arranged in increasing order of probability;  \emph{Row 3:} Probability mass function of trees, arranged in the Signed Distance order of Definition \ref{def:order}, with Fr\'{e}chet mean (red line) of distribution and interquartiles (green dashed lines); \emph{Row 4:} Histogram of distance to the Kingman Fr\'{e}chet mean, with median (blue line) and interquartiles (orange dashed lines) of the distance to the mean; \emph{Row 5:} Histogram of signed distance to the Kingman Fr\'{e}chet mean, with median (blue line) and interquartiles (orange dashed lines) of the signed distance; \emph{Row 6:} Multidimensional scaling visualization of the tree distribution, each dot represents a tree colored by its probability mass, with Fr\'{e}chet mean (red dot) and expected value (green dot).}
	\label{fig:theoretical}
\end{figure}

For ranked tree distributions with $n=100$ leaves, we simulated $N = 1000$ ranked tree shapes with the \texttt{R} package \textsc{apTreeshape} \citep{Maliet2018} and found the $d_{2}$ mean via the simulated annealing of Section \ref{sec:sa}. The resulting means are shown in Figure \ref{fig:frechet_mean_large_n} for $\beta \in \{ -1.9,-1.5,-1,0,100\}$.

\begin{figure}[H]
	\centering
	\includegraphics[width=\textwidth]{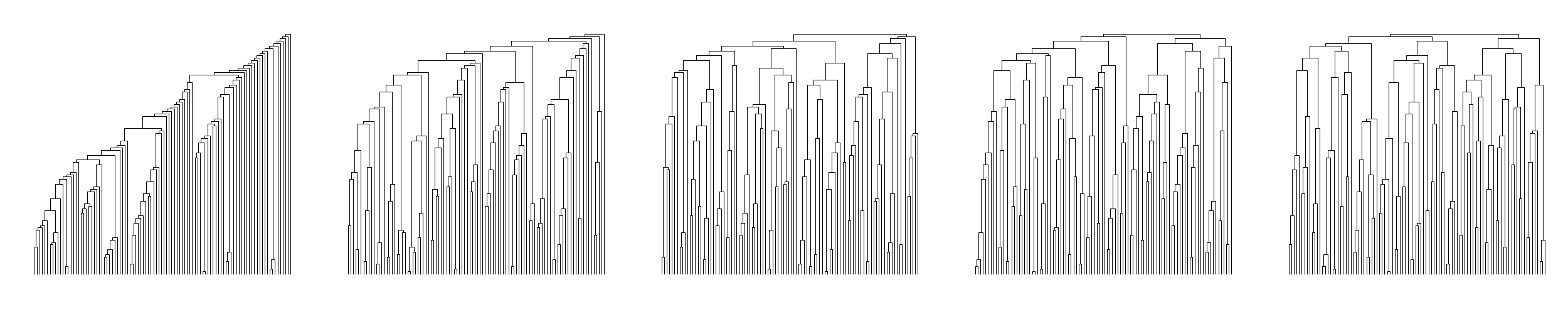}
	\caption{\textbf{Approximated Fr\'{e}chet means.} Fr\'{e}chet means are found via simulated annealing from a sample of $N = 1000$ trees with $n = 100$ leaves from $\beta$-splitting distribution.  Left to right: $\beta = -1.9, -1.5, -1, 0, 100$.  Simulated annealing with exponential cooling schedule and decay parameter $.9995$, initial temperature $1000$.}
	\label{fig:frechet_mean_large_n}
\end{figure}

Figure \ref{fig:theoretical} shows different summaries of four Blum-Fran\c{c}ois distributions on ranked tree shapes with $n=9$ leaves. The second row correspond to the probability mass function with trees (x-axis) arranged in increasing order of probability. When analyzing these plots, it is impossible to assess whether the $\beta=-0.5$ distribution puts more probability mass on unbalanced trees than balanced trees when compared to the $\beta=10$ distribution since the x-axes are not comparable. The third row shows the probability mass functions with the x-axes arranging trees in the signed distance total order. Here, all x-axes correspond to the same tree arrangements. It is now clear that the $\beta=-0.5$ distribution assigns more probability mass to unbalanced trees and the $\beta=10$ distribution assigns more mass to balanced trees. This is confirmed in the fifth row of Figure \ref{fig:theoretical}. The histogram of the signed distance to the mean is skewed to the right (balanced) when $\beta=10$. 

Row four of Figure \ref{fig:theoretical} shows the histograms of the distance to the mean. This one-dimensional summary of the tree distributions hinders whether some distributions put more probability mass to different types of trees. For example the last three histograms of the fourth column look very similar. Finally, while the MDS plots (Figure \ref{fig:theoretical}, last row) only explain about 53\% of all pairwise distances, the last panel shows the distinctions between the four probability mass distributions. Here, green dots corresponds to the points whose $\mathbf{F}$-matrix is $\E(F)$ (and do not lie in tree space) and red dots are the Fr\'{e}chet means. 

We note that the Fr\'{e}chet means of the distributions with $\beta=0.05, 0$ and $10$ are very close to each other. This indicates that a single central summary may not have a good discriminating power for detecting difference in distributions within the Blum-Fran\c{c}ois family.

\subsection{Characterization of mean Kingman tree} \label{ss:theory}

As stated in Section~\ref{ss:mip}, under $d_{2}$, the population Fr\'{e}chet mean is given by
\begin{align*}
	\bar{F}_2 = \underset{F \in \mathcal{F}_n}{\operatorname{argmin}} 
	{\sum_{k,l} \{F_{kl}^2 - 2F_{kl}M_{kl}} \},
\end{align*}
where $M_{kl}=\E(F_{k,l})$. If we know the matrix $M$, we only need to search for ranked tree shapes that are in a neighborhood of $M$ (see for example the red and green dots representing the $M$ and the Fr\'{e}chet mean in the last row of Figure \ref{fig:theoretical}). In fact, the only data input needed using \textsc{gurobi} or simulated annealing is $M$. Although there is no explicit formula for the Fr\'{e}chet mean for the distributions analyzed here, there is a explicit formula for $M$ for the Kingman/Yule coalescent distribution (Blum-Fran\c{c}ois with $\beta=0$). In Figure \ref{fig:theoretical}, we visualize $M$ and the Fr\'{e}chet mean in an MDS plot of the entire space.

\begin{theorem}
	\label{prop:kingman_Fmatrix}
	Let $F \in \mathcal{F}_{n}$ be an $\mathbf{F}$-matrix distributed according to the Blum-Fran\c{c}ois model with $\beta=0$, i.e. according to the Kingman/Yule coalescent distribution, then:
	
	1. The distribution of the $i$-th row of $F$ is independent of $n$.
	
	2. $\E[F_{ij}] = \dfrac{j(j+1)}{i}$ 
	
	3. $\operatorname{Var}[F_{ij}] = \dfrac{j^32(j+1)^2}{i^2(i-1)} + \dfrac{j(j+1)(i-2j-1)}{i(i-1)}$
	
	4. $\operatorname{Cov}[F_{i_1 j_1}, F_{i_2 j_2}] = $
	$$ \begin{cases}
		\frac{j_1(j_1+1)[j_2(j_2+2)+(i_1 + 1)(i_1-2j_2-2)]}{i_1^{2}(i_1 - 1)} &\text{when}~ i_1 = i_2, j_1 < j_2 \\
		\frac{j_2 (j_2 + 1) (j_2-i_2) (j_2-i_2+1)}{i_1 i_2 (i_2-1)} &\text{when}~ i_1 > i_2, j_1 = j_2 \\
		\frac{j_2(j_2+1)[(j_1+1)(j_1+2)+(i_1+1)(i_1-2j_1-2)]}{i_1^{2}(i_1-1)}+\frac{(i_1-i_2)j_1j_2(j_2+1)}{i_1i_2}\left[ \frac{j_1-1}{i_2-1}-\frac{j_1+1}{i_2+1}\right] &\text{when}~i_1 > i_2, j_1 > j_2 \\
		\frac{j_1(j_1+1)[(j_2+1)(j_2+2)+(i_1+1)(i_1-2j_2-2)]}{i_1^{2}(i_1-1)}+\frac{(i_1-i_2)j_1(j_1+1)j_2}{i_1i_2}\left[ \frac{j_2-1}{i_2-1}-\frac{j_2+1}{i_2+1}\right] &\text{when}~i_1 > i_2, j_1 < j_2
	\end{cases}$$
\end{theorem}

The proof can be found in Section \ref{app:proof_kingman_Fmean} 
in the appendix.

The relevance of Theorem~\ref{prop:kingman_Fmatrix} is that for one of the most popular models in population genetics: the coalescent  \citep{wakeley_coalescent_2008}, the Fr\'{e}chet mean can be obtained for any $n$, without the need for simulating a sample from the distribution as it is done in Figure \ref{fig:frechet_mean_large_n}. Moreover, given a sample of $\mathbf{F}$-matrices, the sample average (element-wise) converges almost surely to $M$ by the law of large numbers. Theorem~\ref{prop:kingman_Fmatrix} together with the Multivariate Central Limit Theorem \citep[Thm 5.4.4]{hogg2005introduction}, could be used to test whether a random sample of ranked tree shapes follows the standard coalescent distribution such as the Kingman/Yule model. 

\begin{theorem}(Central Limit Theorem for $\mathbf{F}$-matrices)
	\label{thm:clt}
	Let $F^1, \dots F^m \in \mathcal{F}_n$ be an i.i.d. sample of $\mathbf{F}$-matrices drawn from some distribution $P$.  Let $\bar{F}_m \in \mathbb{R}^{(n-1)\times(n-1)}$ be the matrix whose entries correspond to the sample average of $F^{1},\ldots,F^{m}$ (entrywise).  Then 
	\begin{equation}
		\sqrt{m}\left(\bar{F}_m - M\right) \overset{d}{\rightarrow} N(\mathbf{0}, \mathbf{\Sigma})
	\end{equation}
	where the mean $M \in \mathbb{R}^{(n-1)^2}$ is given by 
	$$\mathbf{\mu}_{ij} = {\E}_P[F_{ij}]$$ 
	and the covariance tensor $\mathbf{\Sigma} \in {\mathbb{R}^{(n-1)^2 \times (n-1)^2}}$ is given by
	$$\mathbf{\Sigma}_{ij,kl} = {\rm Cov}_P [F_{ij}, F_{kl}]$$
\end{theorem}
\begin{proof}
	The proof follows directly from the multivariate central limit theorem, considering the $\mathbf{F}$-matrices as elements of $\mathbb{R}^{(n-1)^2}$.  Since each entry of the $\mathbf{F}$-matrices is bounded in $[0,n]$, all expectations are finite.
\end{proof}

\begin{corollary}
	\label{thm:clt-unknown}
	Consider the setting of Theorem \ref{thm:clt}, and assume that $\mathbf{\Sigma}$ is invertible.  Let ${\hat{\mathbf{\Sigma}}} \in {\mathbb{R}^{(n-1)^2 \times (n-1)^2}}$ be the empirical covariance tensor.  We have
	\begin{equation}
		{\hat{\mathbf{\Sigma}}}^{-1/2}\sqrt{m}\left(\bar{F}_m - \mathbf{\mu}\right) \overset{d}{\rightarrow} N(\mathbf{0}, \mathbf{I})
	\end{equation}
	where $\mathbf{I} \in {\mathbb{R}^{(n-1)^2 \times (n-1)^2}}$ is the identity tensor given by $$\mathbf{I}_{ij,kl} = \mathbf{1}_{i = k, j = l}$$
\end{corollary}
\begin{proof}
	The proof follows by the multivariate version of Slutsky's theorem, using consistency of the empirical covariance, and the fact that the function $M \rightarrow M^{-1/2}$ is continuous when $M$ is an invertible covariance matrix.
\end{proof}

\subsection{Summaries of coalescent ranked genealogical distributions} \label{ss:sim_genealogies}

In neutral (isochronous) coalescent models with variable population size, the tree topology and the distribution of branching event waiting times are independent. The ranked tree shape is distributed according to the Kingman/Yule/Blum-Fran\c{c}ois model with $\beta=0$ and the branching event times have the following conditional density:
\begin{equation}
	f(u_{i-1}\mid u_{i}, N_{e}(t))=\frac{\binom{i}{2}}{N_{e}(u_{i-1})} \exp \left\lbrace-\binom{i}{2}\int^{u_{i-1}}_{u_{i}}\frac{du}{N_{e}(u)}\right\rbrace,
\end{equation}
with $u_{n}=0$ and $N_{e}(t)$ is a non-negative function that denotes the effective population size \citep{Slatkin:1991wx}. 

We simulated 1000 ranked genealogies according to the neutral coalescent model with the following effective population size functions: \\
(1) Constant: $N_{e}(t)=10000$;\\
(2) Exponential: $N_{e}(t)=10000\exp\{-0.01t\}$; \\
(3) Logistic:
\begin{equation*}
	N_{e}(t)= \begin{cases}
		1000+\frac{9000}{1+\exp\left[6-2 (t \mod 12)\right]} & (t \mod 12)\leq 6\\
		1000+\frac{9000}{1+\exp\left[-18+2 (t \mod 12)\right]} & (t \mod 12)> 6.\\
	\end{cases}
\end{equation*}

Figure \ref{fig:frechet_Ne} shows the Multidimensional scaling representation of the three simulated distributions (using $d_{2}$), together with the Fr\'{e}chet means (stars) and medois (triangles). The corresponding means and triangles are depicted in Figure \ref{fig:frechet_Ne2}. Fr\'{e}chet mean topologies are calculated with simulated annealing and the branching event times correspond to the  sample means. Although in this case the two summaries (means and medois) are close to each other in MDS, the two genealogies can have very different branch lenghts as in the logistic simulation (last row, Figure \ref{fig:frechet_Ne2}).

\begin{figure}[H]
	\centering
	\includegraphics[width=\textwidth]{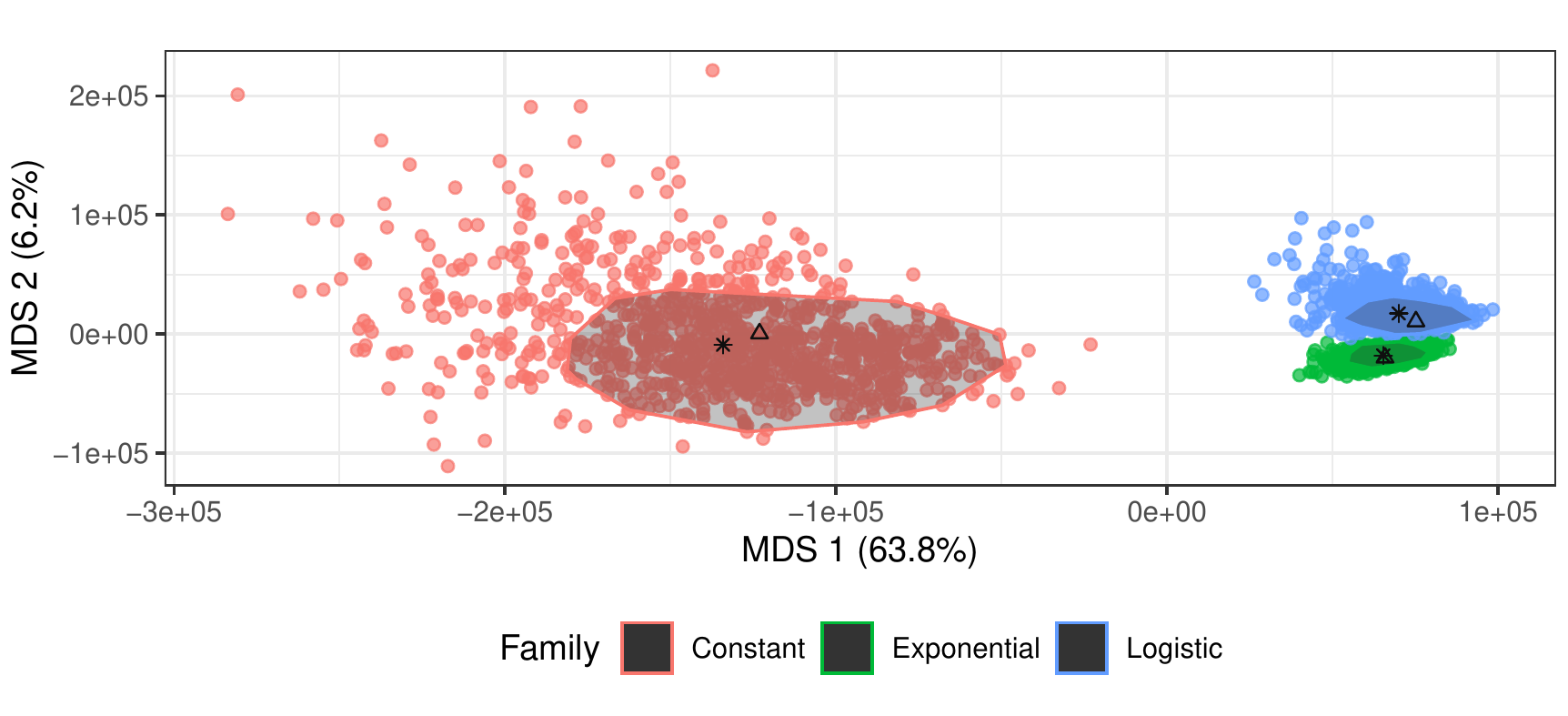}
	\caption{\textbf{MDS of ranked genealogies with different branching event times distributions.} Multidimensional scaling visualization of the coalescent distributions on ranked genealogies with varying effective population sizes: constant (red), exponential (green), and logistic (blue). Triangles denote medoids and stars denote Fre\'{e}chet means. Shaded areas represent the 50\% convex hulls.}
	\label{fig:frechet_Ne}
\end{figure}

\begin{figure}[H]
	\centering
	\includegraphics[width=0.7\textwidth]{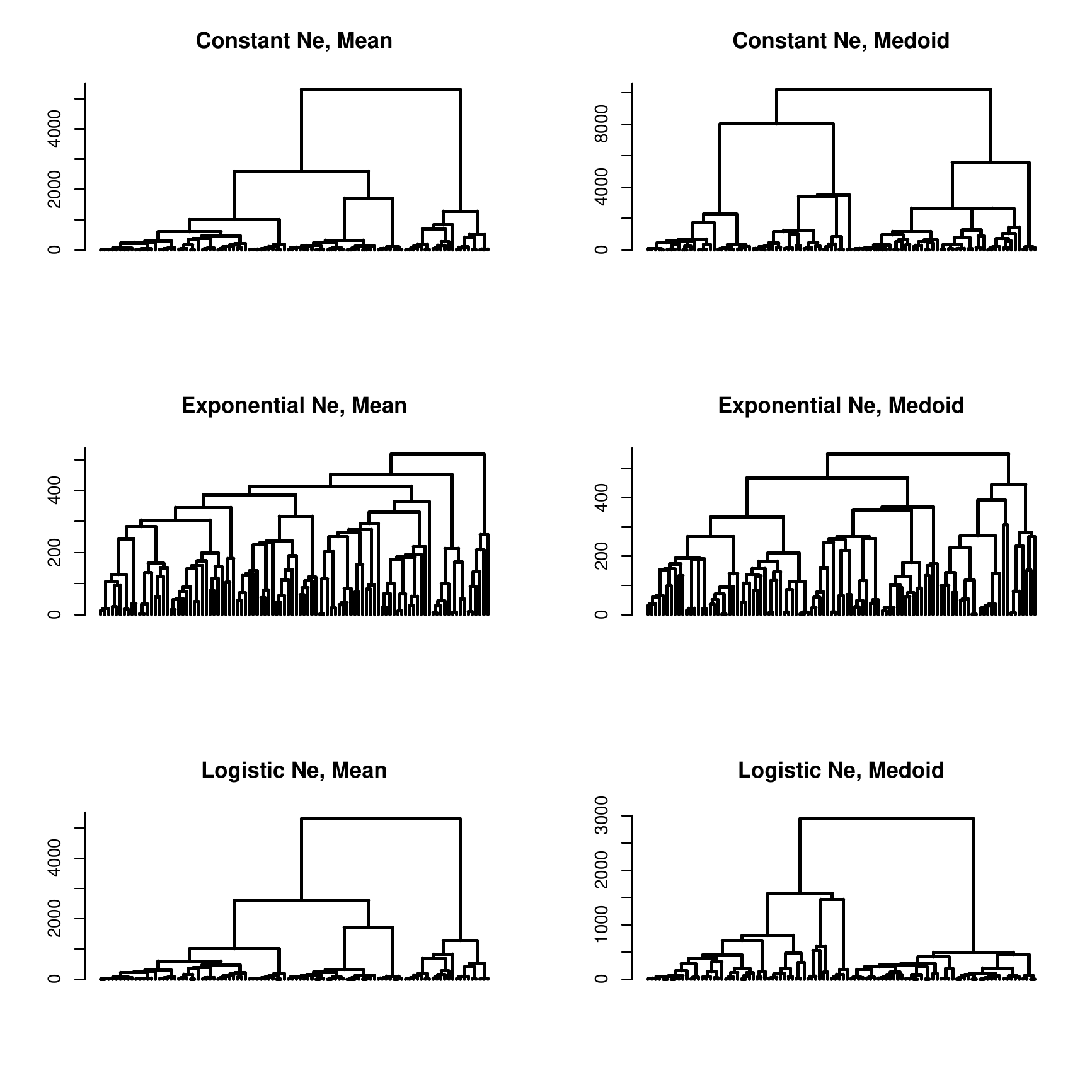}
	\caption{\textbf{Fr\'{e}chet mean vs medoid.} Fr\'{e}chet means (first column) and in-sample medoids (second column) of the three simulated coalescent distributions of genealogies with $n=100$ leaves and with varying effective population size trajectories. Fr\'{e}chet means are indicated as stars and in-sample medoids as triangles in Figure \ref{fig:frechet_Ne}.  Fr\'{e}chet means coalescent times are the sample mean coalescent times.}
	\label{fig:frechet_Ne2}
\end{figure}

\section{Analysis of SARS-CoV-2} \label{ss:covid}

The COVID-19 pandemic has had an enormous impact on all humanity. Tracking the evolution of the SARS-CoV-2 virus that causes COVID-19 has been of great importance for tracking the epidemic and for improving our understanding of the disease \citep{volz2021evaluating}. Here we analyse SARS-CoV-2 molecular sequences publicly available in the GISAID EpiCov database \citep{shu2017gisaid} from the states of California, Florida, Texas and Washington in the USA for the period of February 2020 to September 2020. 

A great challenge in molecular epidemiology of SARS-CoV-2 consists in being able to analyse all available sequences from a population. The number of available sequences exceeds the sample capacity that any Bayesian phylogenetic method (for example those implemented in BEAST \citep{Suchard2018}) can handle. The predominant approach is to subsample the available sequences and infer the posterior genealogical distribution with BEAST.  One important question then is how to assess whether the chosen sample's estimated genealogy is representative to the mean genealogy in the population of a sample of the same size as our subsample. Here, we compare subposterior distributions estimated from different random samples of 100 sequences from California and assess the stability of their Fr\'{e}chet means. 

We selected 9 samples of 100 sequences uniformly at random, obtained in California before June 1 of 2020 and 9 random samples of 100 sequences obtained after June 1, 2020.  For each study (random set of 100 samples), we generated 1,000 MCMC samples thinned every 50,000 iterations from the posterior distribution of model parameters (genealogy  and other parameters) with BEAST \citep{Suchard2018}. Details of parameters and prior distributions selected for BEAST analyses and data access acknowledgments can be found in the appendix. 

We first analyse the California trees sampled from a single posterior distribution. For ease of visualization, we show the MDS plot of all 1000 posterior genealogical samples in Figure \ref{fig:mds_single_Calrandom10} assuming the $d_{2}$ metric defined in Definition \ref{defn:d2-metric}. The posterior central genealogies are depicted as colored dots in Figure \ref{fig:mds_single_Calrandom10} and the corresponding genealogies are drawn in Figure \ref{fig:compare_summaries_Cal}. The two tree topologies of the Fr\'{e}chet means are obtained with our simulated annealing algorithm for heterochronous samples. The Fr\'{e}chet mean coalescent times are obtained as sample averages (purple tree in Figure \ref{fig:mds_single_Calrandom10}) and sample medians (cyan tree in Figure \ref{fig:mds_single_Calrandom10}). Note that the two tree topologies of the Fr\'{e}chet means are different. As mentioned in section \ref{sec:sa}, our SA algorithm finds the tree topology that optimizes the objective function conditioned on a given sequence of sampling and coalescing events.
The maximum clade credibility (MCC) tree is computed using the R package \textsc{phangorn}.  We also plot the in-sample medoid, which is the tree in the sample that minimizes the average distance to the rest of the trees in the sample. We shaded the 50\% credible convex hull around the Fr\'{e}chet mean with mean coalescent times. All central summaries are within the 50\% credible convex hull. 

\begin{figure}[H]
	\centering
	\includegraphics[width=\textwidth]{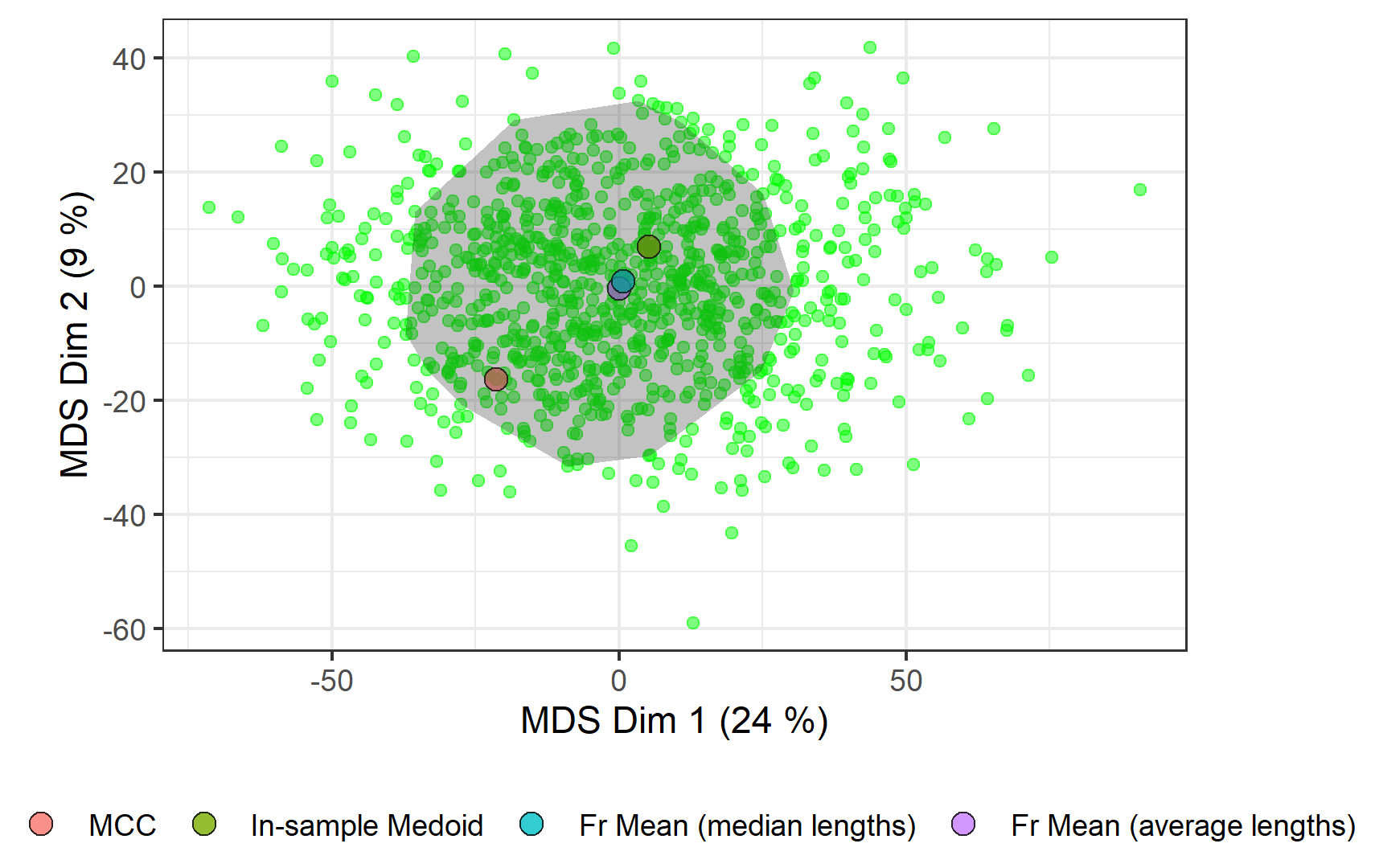}
	\caption{\textbf{MDS plot of a sample of $N = 1000$ trees from a single BEAST posterior distribution of California.}  Each dot represents a tree of $n = 100$ sequences randomly chosen among GISAID sequences from California sequenced between February and May 2020.  Shaded area corresponds to 50\% credible convex hull around the Fr\'{e}chet mean, using average branch lengths. }
	\label{fig:mds_single_Calrandom10}
\end{figure}

\textbf{Remark:} We note that both the Fr\'{e}chet mean and the in-sample medoid are designed to be central with respect to the $d_{2}$ metric using F-matrices, whereas the MCC tree is not.  Hence, it is not surprising that the MCC tree is further away from the center in the MCC plot as compared to the other point summaries.  We can notice that both Fr\'{e}chet means are closer to the center than the in-sample medoid.

\begin{figure}[H]
	\centering
	\includegraphics[width=\textwidth]{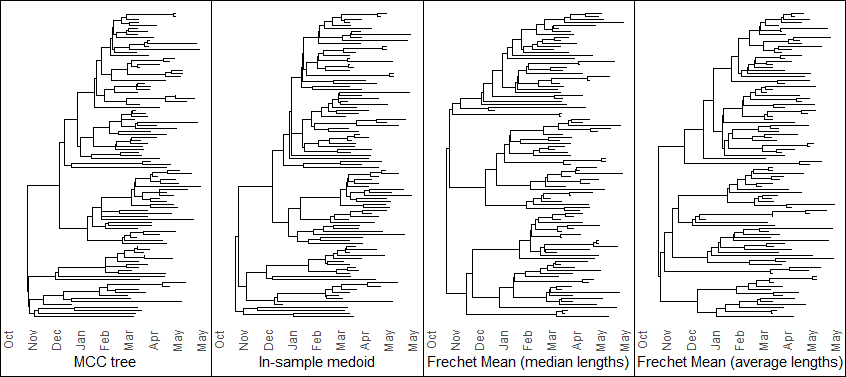}
	\caption{\textbf{Posterior summary trees of a sample of 100 sequences from California Feb-May 2020} Summary trees for a sample of $1000$ trees drawn from a single BEAST posterior using $n=100$ sequences from California. Left to right: Maximum Clade Credibility tree obtained with \textsc{phangorn}, in-sample medoid, Fr\'{e}chet mean using median coalescent times, and Fr\'{e}chet mean using average coalescent times. The four trees correspond to the four dots highlighted in Figure \ref{fig:mds_single_Calrandom10}.}
	\label{fig:compare_summaries_Cal}
\end{figure}

We now visualize multiple subposterior distributions of California together in the MDS plot of Figure \ref{fig:mds_multiple_Indrandom} using the $d_{2}$ metric. For ease of visualization we further subsampled only 20 trees (evenly spaced) from each BEAST posterior distribution to generate the MDS plot. The first group of trees on the left correspond to trees of samples sequenced between February and May 2020, and the second group on the right correspond to trees of samples sequenced between Jun-Sep 2020. The Fr\'{e}chet means are then computed with the SA algorithm for each subsample (shown as bigger dots with black border). 
As expected, the separation between samples from different time periods may reflect the fact that California experienced an increase in COVID-19 prevalence in the second semester of 2020. Moreover, samples from the second time period are more spread. The separation between the Fr\'{e}chet means of the two groups suggests that they could be good statistics in a two-sample (or $k$-sample) test for equality of distributions. We note that even for the first period, some subposterior distributions in California do not overlap.

\begin{figure}[H]
	\centering
	\includegraphics[width=0.8\textwidth]{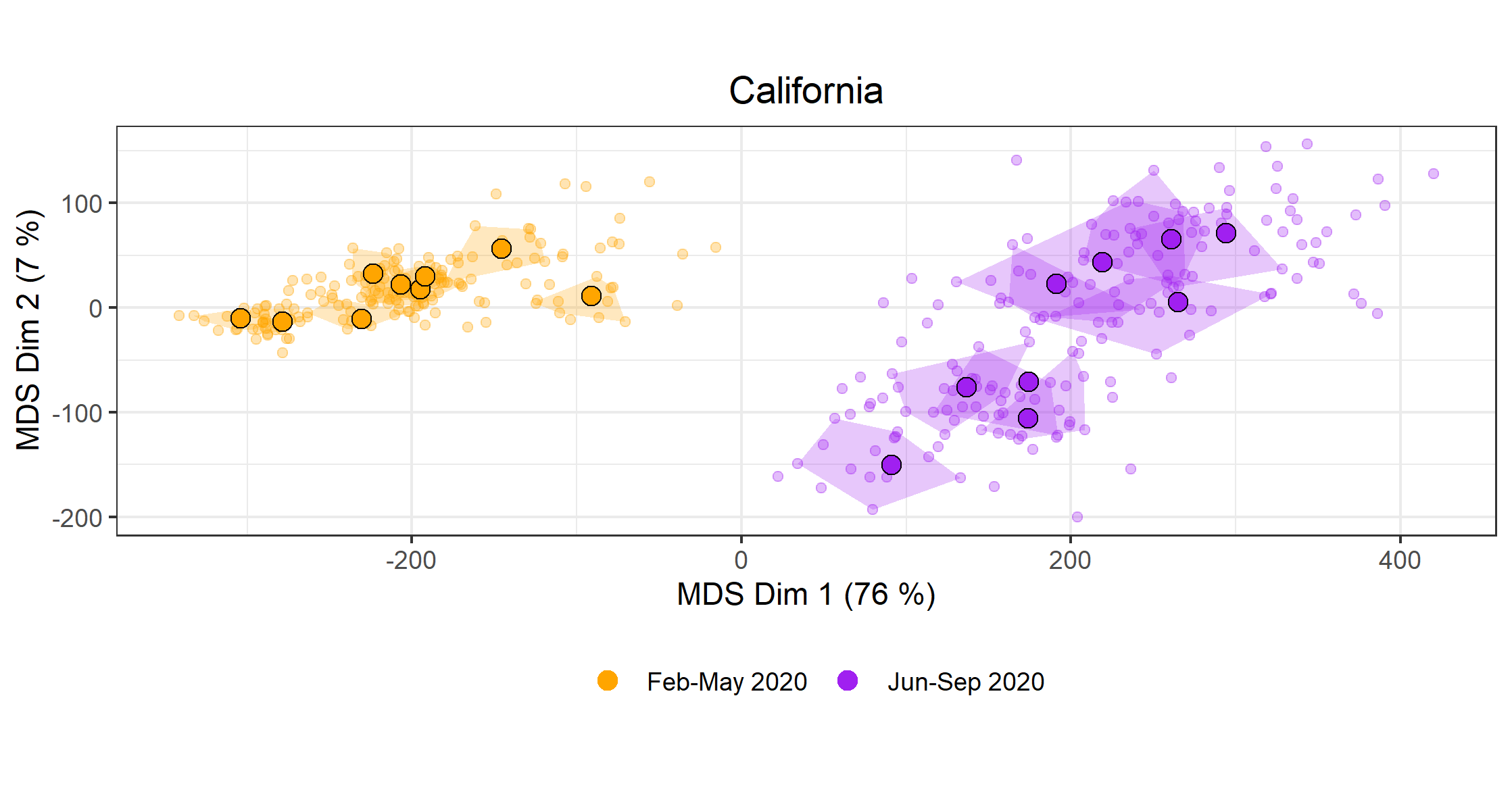}
	\vspace{-1cm}
	\caption{\textbf{MDS plot of 20 trees per each of the 18 posterior BEAST samples of trees with $n=100$ leaves of California}. The group of trees on the left correspond to trees of samples sequenced  Feb-May 2020 and the second group of trees on the right correspond to trees of samples sequenced between Jun-Sep 2020.  The Fr\'{e}chet means are calculated using the average coalescent times and are marked as red dots.}
	\label{fig:mds_multiple_Indrandom}
\end{figure}

We now compare posterior distributions across the four states: California, Washington, Texas and Florida.  Figure \ref{fig:mds_CA_FL_TX_WA} shows the MDS plot of three samples of 20 trees with $n=100$ leaves from each location. Sequences analyzed were genotyped between February and May 2020.  After computing the Fr\'{e}chet means, we downsample in order to better visualize the multidimensional scaling plot. We note that the subposteriors of Washington state are more concentrated than any of the other 3 states. Moreover, there is almost no overlap between the posterior genealogical distribution of Washington state with the posterior distributions of the other states. Indeed the reported number of confirmed COVID-19 cases in Washington state is the lowest and with smallest growth rate of all the states considered in this analysis (Figure \ref{fig:cases_CA_FL_TX_WA}). We note that in the case of Washington state we observe stable subposteriors; all subposteriors reflect the same evolutionary signal.

The posterior distributions of Florida are the second more concentrated in that their three Fr\'{e}chet means are close to each other and their convex hulls overlap. While the posterior distribution of evolutionary trees of Florida are different to those in Washington state, there is some overlap with the posteriors of California and Texas. California experienced the highest cumulative number of cases (Figure \ref{fig:cases_CA_FL_TX_WA}) and it also has the largest heterogeneity in subposterior distributions. Analyses on different subsamples in California can provide very different results. This large heterogeneity may be the result of local outbreaks sequencing efforts in the area.

Finally, while the cumulative number of reported cases in Texas is not as large as the one reported in California. Texas's subposteriors show similar heterogeneity to California, with similar posterior distributions of evolutionary histories (genealogies).

\begin{figure}[H]
	\centering
	\includegraphics[width=\textwidth]{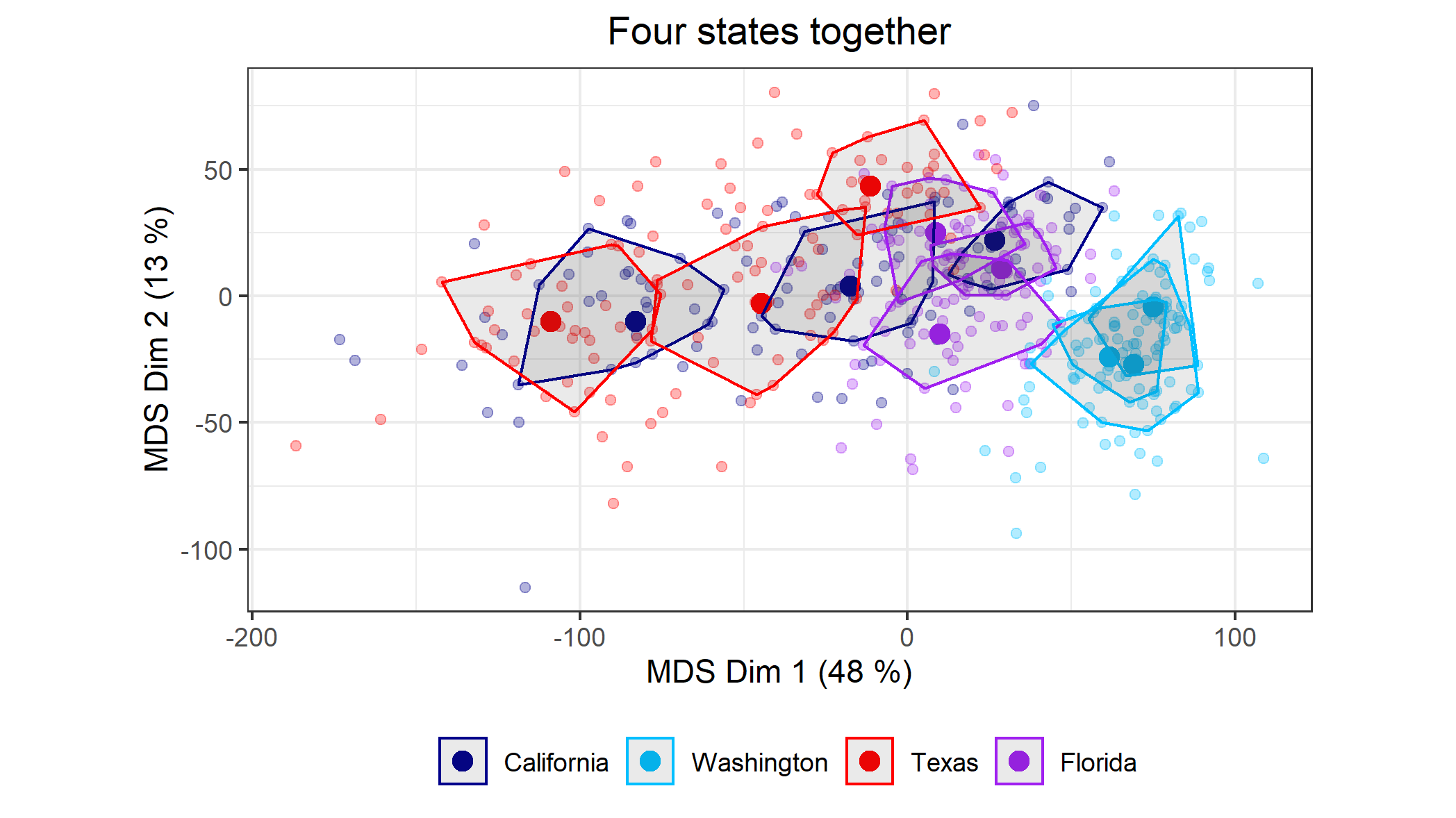}
	\caption{\textbf{MDS plot of multiple samples from California, Washington, Florida, and Texas.} Three samples of 20 trees of $n = 100$ samples randomly chosen among GISAID sequences in Feb-May 2020 per location. The Fr\'{e}chet means are calculated using average coalescent times and marked as red dots  The shaded region corresponds to 50\% credible convex hulls around the Fr\'{e}chet means.}
	\label{fig:mds_CA_FL_TX_WA}
\end{figure}

\begin{figure}[H]
	\centering
	\includegraphics[width=0.4\textwidth]{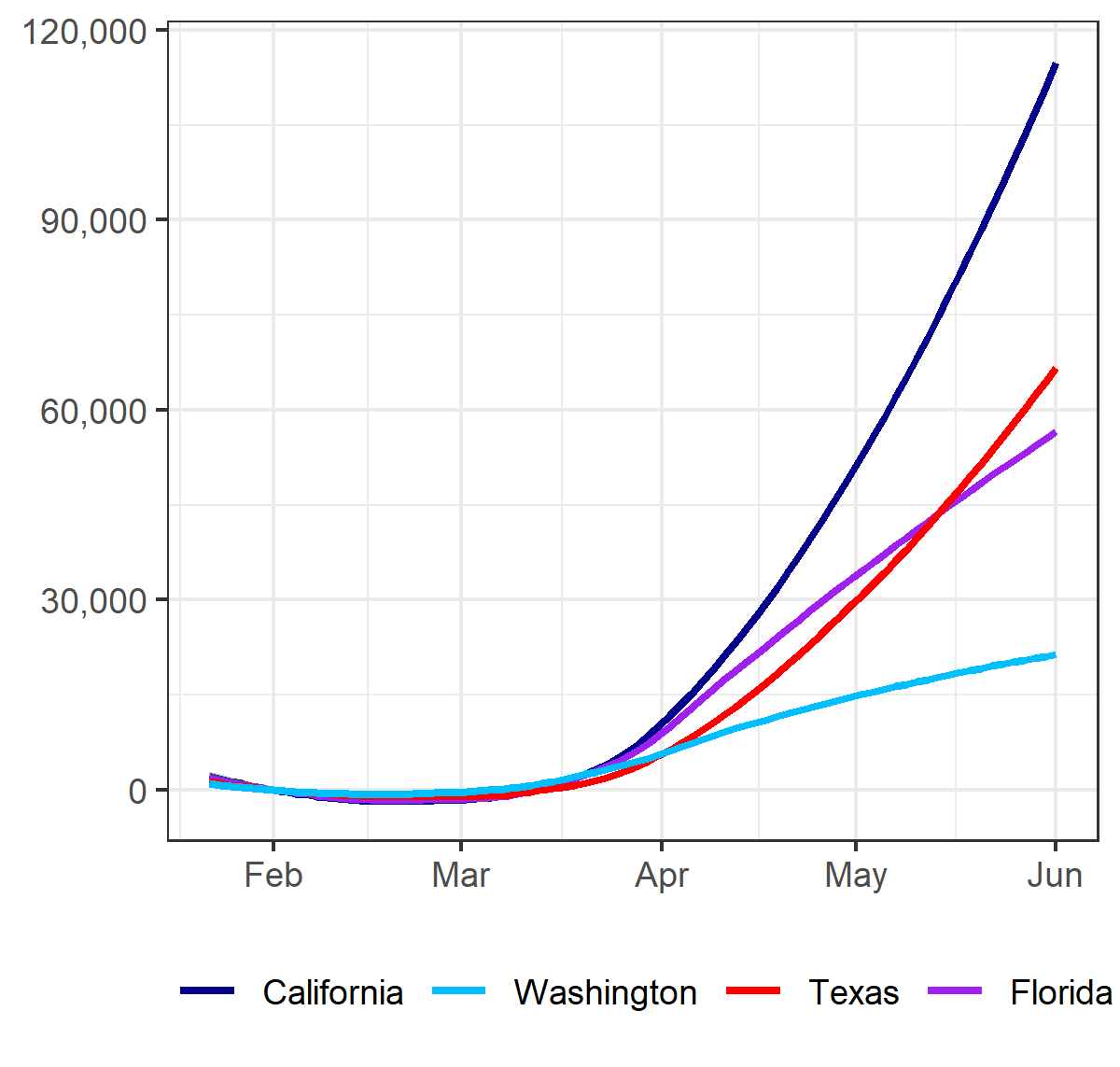}
	\caption{\textbf{Confirmed cumulative case counts in California, Washington, Florida, and Texas.} Data between January and September 2020, obtained from the COVID-19 Data Repository by the Center for Systems Science and Engineering (CSSE) at Johns Hopkins University \citep{jhu-csse}.}
	\label{fig:cases_CA_FL_TX_WA}
\end{figure}


\section{Discussion} \label{sec:discussion}

For discrete tree topologies, the Fr\'{e}chet mean ranked tree shape may not be unique. However in our experience, we found the Fr\'{e}chet means to be very close to each other. We conjecture that the set of Fr\'{e}chet means has a very small diameter and it will be explored in future research. While the non-uniqueness of the Fr\'{e}chet mean can potentially be problematic for hypotheses testing, we remark that the expected $F$ matrix, here denoted by $M$, is unique and the limit of the sample mean $\bar{F}_{m}$ by the strong law of large numbers. In this manuscript, we provided a Central Limit Theorem result for $\bar{F}_{m}$ and analytical expressions for $M$ and variance of the ranked tree shape distribution under the standard Kingman coalescent. Theoretical results of this kind for other distributions is left as future work. Similarly, analyses of several test statistics based on the distances analyzed here and Fr\'{e}chet means such as in \cite{Dubey2019Frechet} are subject of future research.

Extension of current work include defining distances and summaries for multifurcating tree shapes and phylogenetic networks. Our current implementations are publicly available in \verb|github.com/RSmayak/fmatrix|.


\section*{Acknowledgement}
We acknowledge Paromita Dubey and Jaehee Kim for useful discussions. J.A.P. is supported by National Institutes of Health Grant R01-GM-131404 and the Alfred P. Sloan Foundation.


\bibliography{treespace_ref}

\newpage

\appendix

\section{Distance calculation between heterochronous trees with different sampling events} \label{app:additional_events}.

In Section \ref{ss:dist},
we stated the definitions of  metrics for the case of isochronous ranked tree shapes and ranked unlabeled genealogies. Here, we detail the distance calculation for heterochronous trees with the same number of tips.  

Following the notation and details of Section 3 of the appendix of \citet{Kim2019}, let $T^R_1, \dots, T^R_M$ be heterochronous ranked tree shapes with $n$ tips, and $m_1, \dots m_M$ sampling events respectively.  In order to compute pairwise distances, we require these trees to be represented by $\mathbf{F}$-matrices of the same dimension.  We do this by considering the sequence of sampling and coalescent events for each tree, and inserting artificial sampling events to align the sequences of the different trees.  Note that the following formulation is done going backwards in time with time increasing from the present to the past.

Here, we modified the proposed approach to increase computation speed. Instead of inserting artificial sampling events one pair at a time, we insert all artificial sampling events needed by taking all trees at once. That is, for $i = 1,\dots, m$, let $\mathbf{E}^{(i)} = (e^{(i)}_{m_i + n - 1}, \dots, e^{(i)}_{1})$ be the vector of ordered sampling and coalescent events of tree $i$, where $e^{(i)}_{m_i + n - 1}$ denotes the most recent sampling event $(e^{(i)}_{m_i + n - 1} = s)$, assumed to occur at time $u^{(i)}_{m_i + n - 1} = 0$.  In particular, $(e^{(i)}_1 = c)$ denotes the coalescent event at time $u^{(i)}_1$ corresponding to the most recent common ancestor of the leaves in $T^R_i$.  Each $e^{(i)}_j$ is either a sampling event $(e^{(i)}_j = s)$ or a coalescent event $(e^{(i)}_j = c)$.  The event of type $c$ occurs $n-1$ times and the event of type $s$ occurs $m_i$ times in $\mathbf{E}^{(i)}$.

We then transform all $\mathbf{E}^{(i)}$ into extended vectors of higher dimension, in such a way that all vectors are of the same dimension and all coalescent events are position in the same entries for all trees. We first align all $n-1$ coalescent events among all the trees by adding empty spaces when needed.  Once all the type-$c$ events are aligned, we next align the sampling events between two successive coalescent events or between $t=0$ and the first $c$ event.  In each interval, we add additional sampling events to each tree if it has fewer than the maximum number among the different trees.  We denote these by $a$, and assign $0$ new samples to type-$a$ events.  We then construct the $\mathbf{F}$-matrices and their corresponding modified and extended $\mathbf{W}$ matrices following the construction in \citet{Kim2019}.


\section{Proof of Theorem \ref{prop:bijection}} \label{app:proof_bijection}

\begin{proof}
	This proof builds on the Proof of Theorem 1 in the appendix of \citet{Kim2019}.  We specialise to the case of isochronous trees, in which all tips are sampled at time $u_{n} = 0$.
	
	We want to define the mapping between $\mathcal{T}_n$ and $\mathcal{F}_n$.  Consider a given ranked tree shape $T$ with $n$ leaves.  We wish to construct $F$, the corresponding element of $\mathcal{F}_n$, and show that this mapping is a bijection.
	
	Remember that $\mathcal{F}_n$ is the space of $(n-1)\times (n-1)$ $\mathbf{F}$-matrices, which are lower triangular square matrices of non-negative integers that obey the following constraints:
	
	\begin{enumerate}[label=F\arabic*]
		\item The diagonal elements are $F_{i,i} = i+1$ for $i = 1, \dots, n-1$ and the subdiagonal elements are $F_{i+1, i} = i$ for $i = 1, \dots, n-2$. \label{def:Fmatr1}
		\item The elements $F_{i,1}, i=3, \dots, n-1$, in the first column satisfy $\max \{0, F_{i-1, 1} - 1\} \leq F_{i,1} \leq F_{i-1, 1}$. \label{def:Fmatr2}
		\item All the other elements $F_{i,k}, i= 4,\dots, n-1$ and $k = 2, \dots, i-2$ satisfy the following inequalities: 
		\begin{align}
			\max\{0, F_{i, k - 1}\} &\leq F_{i,k} \label{defF3:1}\\
			F_{i-1, k} - 1 &\leq F_{i,k} \leq F_{i-1, k}  \label{defF3:2} \\
			F_{i, k - 1} + F_{i-1, k} - F_{i-1,k-1} - 1 &\leq F_{i,k} \leq F_{i,k-1} + F_{i-1, k} -F_{i-1,k-1} \label{defF3:3}
		\end{align} \label{def:Fmatr3}
	\end{enumerate}
	
	We first set up some notation.  We view the tree as a branching process starting at the root, and ending at the time when there are $n$ leaves. Let each internal node of the ranked tree shape be labeled with the number of branches in the tree immediately after it bifurcates, and let $u_{i-1}$ be the time at which node $i$ bifurcates.  I.e. The root has label $2$, and the last internal node has label $n$.  Let $I_i$ be the interval between $u_{i}$ and $u_{i+1}$, i.e. the interval between the bifurcations of nodes $i-1$ and $i$.
	
	We define an intermediate matrix $D$ of the same shape before defining $F$.  Let $D$ be an $(n-1) \times (n-1)$ lower triangular matrix, where $D_{ij}$ is the number of descendants of node $j+1$ (i.e. that were born at time $u_{j}$), that have not yet bifurcated until $u_{i+1}$, i.e. the end of interval $I_i$.  We see that $D_{ij} \in \{0, 1, 2\}$, and in particular, we must also have:
	
	\begin{enumerate}[label=D\arabic*]
		\item $D_{ii} = 2, ~ i = 1, \dots {n-1}.$ 		\label{def:Dmatr1}
		\item $D_{(i-1)j} - 1 \leq D_{ij} \leq D_{(i-1)j} ~ \forall ~ i, j.$ \label{def:Dmatr2}
		\item Exactly one of the branches bifurcates at time $u_{i}, i = 1, \dots, n-1$, i.e. $\forall ~ j$, 
		\begin{equation*}
			\sum_{j=1}^{i-1}{\mathbf{1}_{\{ D_{(i-1)j} - D_{ij} = 1 \}}} = 1,
		\end{equation*}
		where $\mathbf{1}_P$ is the indicator function, which is $1$ if $P$ is true and $0$ otherwise. \label{def:Dmatr3}
	\end{enumerate}
	
	We see that this matrix $D$ is uniquely determined for any given binary ranked tree shape $T$.  We can reconstruct $T$ given the matrix $D$ by starting at the root and successively bifurcating branches according to the column of $D$ in which we observe a difference between $D_{(i-1)j}$ and $D_{ij}$.  This shows that the space $\mathcal{T}_n$ is in bijection with $\mathcal{D}_n$, the space of matrices $D$ that obey the conditions above.  We now show that this space is in bijection with $\mathcal{F}_n$.  
	
	For a given matrix $D$, let $\varphi(D) = F$ be given by 
	$$F_{ij} = \sum_{k=1}^j {D_{ik}}.$$
	The inverse mapping is $\varphi^{-1}:\mathcal{F}_n \rightarrow \mathcal{D}_n$ which is given by 
	$$D_{ij} = F_{ij} - F_{i(j-1)}$$
	considering $F_{i0} = 0$.   
	
	In terms of $T$, $F_{ij}$ denotes the total number of branches  at time $u_{j}$ that have not bifurcated by time $u_{i+1}$, i.e. the end of interval $I_i$.  It is clear that the mapping $\varphi: D \mapsto F$ is a one-to-one mapping.  We are left to show that the range of the mapping $\varphi$ is exactly $\mathcal{F}_n$.  It suffices to show that $\varphi(\mathcal{D}_n) \subseteq \mathcal{F}_n$ and $\varphi^{-1}(\mathcal{F}_n) \subseteq \mathcal{D}_n$.
	
	Consider a matrix $D \in \mathcal{D}_n$, and $F = \varphi(D)$.
	After any branching event, there is one node that has split, and the rest are as they were.  This translates to
	$ \sum_{k=1}^{i-1}{D_{ik}}= \left(\sum_{k=1}^{i-1}{D_{(i-1)k}}\right) - 1 $.  Note $F$ must satisfy $F_{ii} = i+1$, as $F_{11} = D_{11} = 2$, and 
	\begin{align*}
		F_{ii} = \sum_{k=1}^{i}{D_{ik}} 
		= D_{ii} + \sum_{k=1}^{i-1}{D_{ik}} 
		= 2 + \left(\sum_{k=1}^{i-1}{D_{(i-1)k}}\right) - 1 
		= F_{(i-1)(i-1)} + 1
	\end{align*}
	
	We will also have $F_{(i+1)i} = i$, as
	\begin{align*}
		F_{(i+1)i} = \sum_{k=1}^{i}{D_{(i+1)k}} 
		= \left(\sum_{k=1}^{i}{D_{ik}}\right) - 1 
		= F_{ii} - 1 = i
	\end{align*}
	
	This shows that $F$ satisfies condition \ref{def:Fmatr1}.
	
	We have $F_{i1} = D_{i1} ~ \forall ~ i$.  This shows condition \ref{def:Fmatr2} using \ref{def:Dmatr2} for $j=1$.
	
	Now we come to condition \ref{def:Fmatr3}.  \ref{def:Fmatr3}.\ref{defF3:1} follows by non-negativity and the face that the columns of $D$ are monotonically decreasing.  \ref{def:Fmatr3}.\ref{defF3:2} follows using \ref{def:Dmatr3}. \ref{def:Fmatr3}.\ref{defF3:3} is equivalent to 
	\begin{align*}
		0 &\leq  (F_{(i-1)k} - F_{ik}) - (F_{(i-1)(k-1)} - F_{i(k-1)})  \leq 1 \\
		\iff 0 &\leq (\sum_{j=1}^k{D_{(i-1)j}} - \sum_{j=1}^k{D_{ij}}) - (\sum_{j=1}^{k-1}{D_{(i-1)j}} - \sum_{j=1}^{k-1}{D_{ij}})  \leq 1 \\
		\iff 0 & \leq D_{(i-1)k} - D_{ik} \leq 1
	\end{align*}
	which follows from \ref{def:Dmatr2} and \ref{def:Dmatr3}.
	This shows that $\varphi(\mathcal{D}_n) \subseteq \mathcal{F}_n$.
	
	Now let $F$ be an element of $\mathcal{F}_n$ and $D = \varphi^{-1}(F)$.
	
	We first see that $D_{ii} = F_{ii} - F_{(i-1)i} = (i+1) - (i-1) = 2$.  This shows \ref{def:Dmatr1}.  Following the case above, we see that \ref{def:Fmatr3}.\ref{defF3:3} implies \ref{def:Dmatr2}, and \ref{def:Fmatr3}.\ref{defF3:2} implies \ref{def:Dmatr3}. This shows that $\varphi^{-1}(\mathcal{F}_n) \subseteq \mathcal{D}_n$ and completes the proof of the bijection.
	
\end{proof}


\section{Proof of Proposition \ref{prop:separate_top_time}} \label{app:proof_separate_top_time}

\begin{proof}
	Let $G = (F^G, \mathbf{u}^G)$, $H = (F, \mathbf{u}) \in \mathcal{G}_{n}$, two genealogies.
	Let $\mathbf{u}^{(.)} = (u_1^{(.)}, \dots, u_{n-1}^{(.)})$ be the vector of branching event times for $(.) = G, H$. For convenience, the weight matrix 
	$$w(\mathbf{u}^{(.)}) = W^{(.)} = 
	\begin{pmatrix}
		u_1^{(.)} - u_2^{(.)} \\
		u_1^{(.)} - u_3^{(.)} & u_2^{(.)} - u_3^{(.)} \\
		u_1^{(.)} - u_4^{(.)} & u_2^{(.)} - u_4^{(.)} & u_3^{(.)} - u_4^{(.)} \\
		\vdots & & & \ddots \\
		u_{1}^{(.)} & u_{2}^{(.)} & \dots & \dots &  u_{n-1}^{(.)}
	\end{pmatrix}, $$
	is vectorized as $w_1, \dots w_m$, $m = \dfrac{n(n-1)}{2}$. We then can re-express $d_{2}$ as $d_2(G, H)^2 := \sum_{j=1}^m{(F^G_{j}W^G_{j} - F_{j}W_{j})^{2}}$. Further, we assume that under $\nu$, the tree topology and the coalescent times are independent, that is $d\nu(H)=\mu(F)\prod^{n-1}_{j=1}f(u_{j} \mid u_{j+1})d(\mathbf{u}) = \mu(F)f(\mathbf{u})d(\mathbf{u})$ and re-express Eq.~\ref{eq:frech_gen}
	for $d_{2}$ as follows:
	\begin{equation*}
		\bar{G}_{2} \in 
		\argmin_{G=(F^G,u^G) \in \mathcal{G}_{n}} 
		\sum_{F \in \mathcal{F}_{n}} 
		\underset{0 = u_n < u_{n-1} < \dots < u_2 < u_1} {\int}
		\sum_{j=1}^m{(F^G_{j}W^G_{j} - F_{j}W_{j})^{2}}
		\mu(F) 
		f(\mathbf{u})d(\mathbf{u}).
	\end{equation*}
	
	Let $I^{G}_k = u^{G}_k - u^{G}_{k+1}$
	with $u^{G}_n = 0$, and let $I^{G} = (I^{G}_1, \dots, I^{G}_{n-1})$, then the weight matrix of $W^{G}$ becomes 
	$$g(I^{G}) = W^{G} = \begin{pmatrix}
		I^{G}_1 \\
		I^{G}_1 + I^{G}_2 & I^{G}_2 \\
		I^{G}_1 + I^{G}_2 + I^{G}_3 & I^{G}_2 + I^{G}_3 & I^{G}_3 \\
		\vdots & & & \ddots \\
		I^{G}_1 + \cdots + I^{G}_{n-1} & I^{G}_2 + \cdots + I^{G}_{n-1} & \dots & \dots &  I^{G}_{n-1}
	\end{pmatrix}, $$
	and 
	
	\begin{equation}
		A(G) = A(F^G, I^G) =
		\sum_{F \in \mathcal{F}_{n}} {
			\underset{0 = u_n < u_{n-1} < \dots < u_2 < u_1} {\int}{
				\sum_{j=1}^m {
					(F_j w_j(\mathbf{u}) - F^G_j g_j(I^G))^2 \mu(F) f(\mathbf{u}) d(\mathbf{u})
			}}
		}
		\label{min-objective}
	\end{equation}
	
	To find	the branching event time intervals $I^{G}$ that minimize $A(G)$, we take the partial derivative of $A(G)$, a continuous function of $I^G$, with respect to $I^G_k$. We exchange the order of summation and integration since everything is positive and exchanging differentiation and integration we get 
	
	\begin{align}
		\dfrac{\partial A}{\partial I^G_k}
		= 0 = 
		&\sum_{F \in \mathcal{F}_{n}} {
			\underset{0 = u_n < u_{n-1} < \dots < u_2 < u_1} {\int}{
				\sum_{j=1}^m {
					2 \cdot (F_j w_j(\mathbf{u}) - F^G_j g_j(I^G))\cdot F^G_j \cdot \dfrac{\partial g_j}{\partial I^G_k}(I^G) \mu(F)f(\mathbf{u}) d(\mathbf{u}).
			}}
		}  \\
		\Rightarrow
		&\sum_{F \in \mathcal{F}_{n}} {
			\sum_{j=1}^m {
				\left\{F_j^{H} \cdot \E[w_j(\mathbf{u})]\cdot F^G_j \cdot \dfrac{\partial g_j}{\partial I^G_k}(I^G) \cdot\mu(F) - (F^G_j)^2 \cdot g_j(I^G) \cdot \dfrac{\partial g_j}{\partial I^G_k}(I^G) \cdot \mu(F) \right\}
			}
		} = 0. \\
		\Rightarrow &\sum_{j=1}^m \left\{
		F^G_j \cdot \E[w_j(\mathbf{u})] \cdot \dfrac{\partial g_j}{\partial I^G_k}(I^G) \cdot 
		\sum_{F \in \mathcal{F}_{n}} \mu(F) F_j
		\right\}
		= 
		\sum_{j=1}^m {
			F^G_j \cdot g_j(I^G) \cdot \dfrac{\partial g_j}{\partial I_k}(I^G) \cdot 
			F^G_j .
		}
		\label{final-eqn}
	\end{align}
	
	Here the expectation is taken with respect to the joint density of $u$, i.e. $\E[w_j(\mathbf{u})] = \int_{\mathbf{u}} w_j(\mathbf{u}) f(\mathbf{u}) d(\mathbf{u})$.
	We see that (\ref{final-eqn}) is satisfied simultaneously for all $k$ if we have:
	$$F^G_j= \sum_{F \in \mathcal{F}_{n}} \mu(F) F_j=\E(F_{j}),$$ and  
	$$g_j(I^G)= \E[w_j(\mathbf{u})],$$
	for $j = 1, \dots m = n(n-1)/2$. We note that the expected value $\E(F)$ may not correspond to the Fr\'{e}chet mean tree;  although this is a solution, we are ignoring the constraints imposed to $F$.
\end{proof}


\section{A Markov chain on the space of ranked tree shapes} \label{app:mc-encod}

We drop the F matrix representation of ranked tree shapes and instead use two string representations of the spaces of isochronous and heterochronous ranked tree shapes respectively. We use the string representations to define two Markov chains on the corresponding spaces.  An isochronous ranked tree shape is encoded as a string of $n-1$ integers $t=(t_{1},t_{2},\ldots,t_{n-1})$, where $t_k$ indicates the parent node of the internal node with ranking $k+1$, $k \in \{1,\ldots,n-1\}$. It is assumed that the first integer $t_{1}$ of the string representation is 1 (parent of root node). 
Figure \ref{fig:fmat5} 
shows the string encodings of each of the $5$ ranked tree shapes at the bottom. The string representation was introduced earlier as the functional code for binary increasing trees \citep{donaghey1975alternating}. The set of all string representations of $n-1$ elements are in bijection with the space of isochronous ranked tree shapes of $n$ leaves and the space of binary increasing trees of $n-1$ nodes \citep{richard1999enumerative}.

To recover the tree $T$ from the encoding $t$, we can proceed in a generative fashion: we start at the root which has label $2$, and proceed by bifurcating the leaves in the order determined by $t$.  The space of strings $\mathfrak{T}_n$ is the set of all $t$ strings of length $n-1$ defined as follows:

\begin{definition}(Isochronous string representation). \label{encod-defining_props}
	A string $t$ of non-negative integers that encodes an isochronous ranked tree shape has the following defining properties:
	
	1. $t_{1}=1$
	
	2. For $i > 1$, $2 \leq t_i \leq i$.
	
	3. No entry of $t$ can appear more than twice.
\end{definition}

\begin{definition}(Markov chain on isochronous strings). \label{def:mc_defn}
	Let $t \in \mathfrak{T}_n$ be a string encoding an isochronous ranked tree shape as described in Definition \ref{encod-defining_props}.  We define a Markov chain on $\mathfrak{T}_n$ as follows:
	
	1. Pick an element $i \in 2, \dots, n$ uniformly at random.
	
	2. Pick the value of $t_i$ uniformly at random from the allowable choices in $2, \dots, i$, i.e. from those choices that do not already appear twice among $t_{-i}$.
	
\end{definition}

\begin{proposition} The Markov chain on isochronous strings (Definition \ref{def:mc_defn}) is ergodic with uniform stationary distribution on the space of strings of length $n-1$, or equivalently, on the space of ranked tree shapes with $n$ leaves.
\end{proposition}

\begin{proof}
	Let $t$ be an arbitrary element of $\mathfrak{T}_n$. We show that $t$ is path connected to $t^* = (1,2,\dots,n-1)$.  This string corresponds to the most unbalanced tree, also called the caterpillar or the comb tree.  Since the Markov chain is symmetric, $t^*$ is path connected to every element of $\mathfrak{T}_n$ as well and hence the chain is irreducible. The following path has all transitions with positive probability:
	
	\begin{align*}
		t = t^{(0)} &= (1, t_2, t_3, \dots, t_{n-2}, t_{n-1}) \\
		t^{(1)} &= (1, t_2, t_3, \dots, t_{n-2}, n-1) \\
		t^{(2)} &= (1, t_2, t_3, \dots, n-2, n-1) \\
		&\vdots \\
		t^{(n-3)} &= (1, t_2, 3, \dots, n-2, n-1) \\
		t^{(n-2)} &= (1, 2, 3, \dots, n-2, n-1) = t^*
	\end{align*}
	Note that $t^{(i)} \mapsto t^{(i+1)}$ is always a valid transition due to Property \ref{encod-defining_props}.3.
\end{proof}


\begin{figure}
	\centering
	\includegraphics[width=0.6\textwidth]{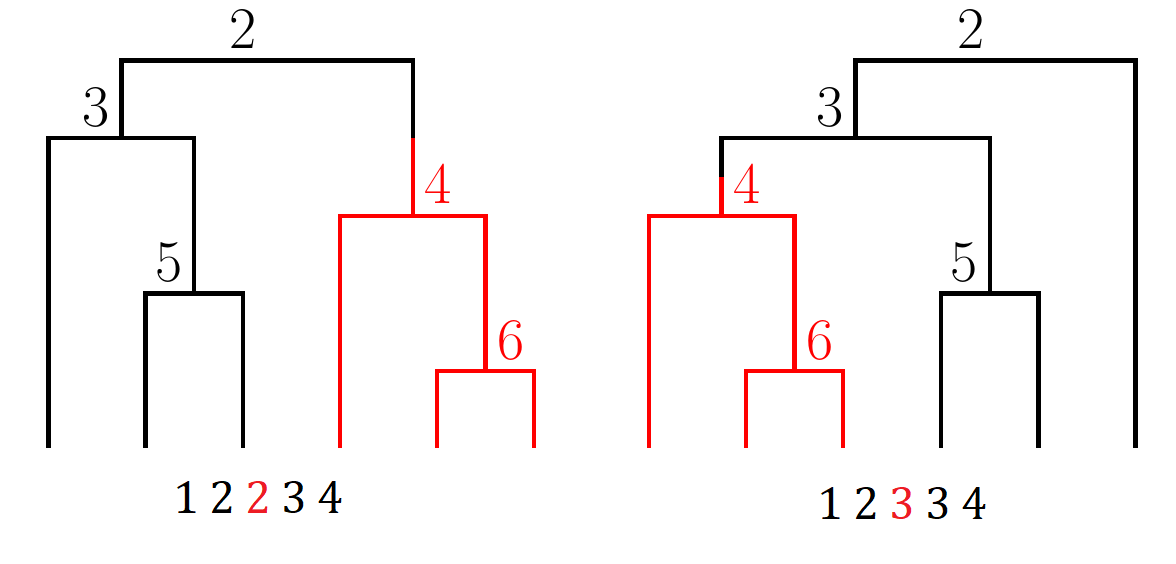}
	\caption{An example transition under the Markov chain of Definition \ref{def:mc_defn} from $(1, 2, 2, 3, 4)$ to $(1, 2, 3, 3, 4)$.  The subtree of node $4$ is plucked from under node $2$ and planted under node $3$.
	}
	\label{fig:example_transition}
\end{figure}

This representation can be extended to heterochronous trees as well, with additional entries indicating the sampling events.  We define these strings in the following way:

\begin{definition}(Heterochronous string representation). \label{encod-defining_props_hetero}
	A heterochronous ranked tree shape with $n$ leaves is encoded as a pair of strings $(t, \sigma)$ each of length $2n-1$. As before, $t$ is a string of non-negative integers that indicates the parent nodes of internal nodes (coalescent events), however, $t$ now also includes the parent nodes of all leaves. The sequence order is given by the time they are created and $\sigma$ is a $0-1$ string that indicates whether the corresponding node is internal (1) or a leaf (0).  These strings have the following defining properties:
	
	1. $t_{1}=1$, $\sigma_1 = 1$
	
	2. $|\{i:\sigma_i = 1\}| = n-1$
	
	3. $|\{i:\sigma_i = 0\}| = n$
	
	4. Each element of $\{2,\dots n\}$ occurs exactly twice in $t$.
	
	5. For each $i > 1$, $2 \leq t_i \leq 1 + \sum_{j=1}^{i-1}{\sigma_j}$.
	We note that the string $t_\sigma = (t_i:\sigma_i = 1)$ is a valid string encoding an isochronous ranked tree shape.
\end{definition}

For example, the string representation of the ranked tree shape of Figure \ref{fig:rut_example}(C) 
is $(t=123442365567788$, $\sigma=111100101100000)$. We note that the string $t_\sigma = (t_i:\sigma_i = 1)$ is a valid string encoding of an isochronous ranked tree shape. In addition, this representation can also admit extensions to multifurcating trees, which we leave for future study.

While defining Markov chains on the space of ranked tree shapes is useful for Bayesian inference in areas such as phylogenetics and phylodynamics, we rely on these chains for finding the Fr\'{e}chet mean via stochastic combinatorial optimization (Section \ref{sec:sa}).

\begin{definition}(Markov chain on heterochronous strings). \label{def:mc_defn_hetero}
	Let $(t, \sigma)$ be a pair of strings encoding an heterochronous ranked tree shape as described in Property \ref{encod-defining_props_hetero}.  We define a Markov chain on the space of such strings, conditional on $\sigma$ a fixed sequence of sampling and coalescent events, with transitions as follows:
	
	1. Pick two distinct element $i, j \in 2, \dots, 2n-1$ uniformly at random.
	2. Swap $t_i$ and $t_j$.
	If the result is a valid heterochronous string, accept the move, otherwise reject the move.
\end{definition}

The Markov chain of Definition \ref{def:mc_defn_hetero} on the space of ranked tree shapes with a given $\sigma$, i.e., with a given sequence of sampling and coalescence events, is also symmetric, aperiodic (we can pick a pair with the same label with positive probability) and irreducible.  The proof of irreducibility is similar to the isochronous case by considering only the coalescent events.

\begin{proof}
	Let $(t, \sigma)$ be the encoding of an arbitrary heterochronous ranked tree shape.   
	
	We first define $(t^*, \sigma)$ which is the analogue of the caterpillar or the most unbalanced tree, but with the given $\sigma$.  Let $t^*_\sigma = (t^*_i:\sigma_i = 1)$ be equal to $(1,2,\dots,n-1)$, and $t^*_{-\sigma} = (t^*_i:\sigma_i = 0)$ be equal to $(2, 3,\dots, n-1, n, n)$.  
	We can follow a similar method as in the isochronous case and show that $t$ is path connected to $t^*$ by a sequence of steps with positive probability.  Note that $t$ has length $2n-1$.
	
	Let $t^{(0)} = t$.  We obtain $t^{(i)}$ by swapping two terms of $t^{(i-1)}$ such that the last $i$ terms of $t^{(i)}$ and $t^*$ are equal.  Explicitly, let $j$ be the largest element in $2, \dots, 2n-i$ such that $t^{(i-1)}_j = t^*_{2n-i}$.  We can swap $t^{(i-1)}_j$ and $t^{(i-1)}_{2n-i}$, since $t^*_{2n-i}$ is the maximum allowable entry at position $2n-i$, and hence $t^{(i-1)}_{2n-i} \leq t^{(i-1)}_j$ which satisfies Definition \ref{encod-defining_props_hetero}.5.  The remaining conditions under Definition \ref{encod-defining_props_hetero} are not affected by the transitions.  Iteratively, we obtain $t^{(2n-1)} = t^*$. 
	
	Since the Markov chain is symmetric, we see that $t^*$ will be path connected to every $t$ as well, and the Markov chain is irreducible. 
	
\end{proof}


\section{Proof of Theorem \ref{prop:kingman_Fmatrix}} \label{app:proof_kingman_Fmean}
\begin{proof}
	\begin{enumerate}
		\item  This follows by the Markovian property of the Yule model, since the $i$-th row of the $\mathbf{F}$-matrix is only determined by the state of the tree when it has $i+1$ tips.
		
		\item  View $P$ as the Yule model, in which we start with one leaf, and successively bifurcate an independently choesn leaf at random.  Let $D_{i,j}$ be the number of branches that were created after the $j$-th split (i.e. when the node with label $(j+1)$ bifurcates), that have not bifurcated by interval $i$.  We must have $D_{ij} \in \{0, 1, 2\}$, and $F_{ij} = \sum_{k\leq j}{D_{ik}}.$
		
		We have $D_{11} = F_{11} = 2$, and $D_{i1} = F_{i1} \in \{0, 1\}$ for $i > 1$. 
		
		For the children of node $2$, i.e. when $j=1$, $D_{ij} = 1$ means that the extant branch descending from node $2$ has not bifurcated yet by interval $i$.  Since at each step in the Yule model, the branch to split is chosen at random, we have
		
		\begin{equation}
			P(D_{i1} = 1) = \dfrac{2}{3} \times \dfrac{3}{4} \times \cdots \times \dfrac{i-2}{i-1} \times \dfrac{i-1}{i} = \dfrac{2}{i}
		\end{equation}
		
		This implies 
		\begin{equation}
			E[D_{i1}] = P(D_{i1} = 1) = \dfrac{2}{i}
		\end{equation}
		
		Let $j > 1$.  Then $D_{ij} = 2$ means that neither of the branches descending from node $(j+1)$ have bifurcated by interval $i$.  By the same argument as above,
		
		\begin{equation}
			P(D_{ij} = 2) = \dfrac{j-1}{j+1} \times \dfrac{j}{j+2} \times \cdots \times \dfrac{i-3}{i-1} \times \dfrac{i-2}{i} = \dfrac{j(j-1)}{i(i-1)}
		\end{equation}
		
		Now $D_{ij} = 1$ means that exactly one branch descending from node $(j+1)$ has bifurcated by interval $i$.  We split this up according to the interval where this bifurcation occurs:
		
		\begin{align}
			P(D_{ij} = 1) &= \sum_{k = 1}^{i-j} {P(D_{ij} = 1, D_{(j+k-1) j} - D_{(j+k)j} = 1)}  \nonumber\\
			&=\sum_{k = 1}^{i-j} {\left(\prod_{l = 1}^{k-1}{\dfrac{j+l-2}{j+l}} \right) \times \dfrac{2}{j+k} \times \left(\prod_{l = k+1}^{i}{\dfrac{j+l-1}{j+l}} \right)} \nonumber\\
			&=\sum_{k = 1}^{i-j} {\dfrac{2j(j-1)}{i} \dfrac{1}{(j+k-2)(j+k-1)}} \nonumber\\
			&= \dfrac{2j(i-j)}{i(i-1)}
		\end{align}
		
		We thus have 
		\begin{align}
			P(D_{ij} = 0) &= 1 - P(D_{ij} = 1) - P(D_{ij} = 2) \nonumber \\
			&= \dfrac{(i-j)(i-j-1)}{i(i-1)}
		\end{align}
		
		and 
		\begin{align}
			E[D_{ij}] &= P(D_{ij} = 1) + 2P(D_{ij} = 2) \nonumber \\
			&= \dfrac{2j(i-j)}{i(i-1)} + 2\dfrac{j(j-1)}{i(i-1)} \nonumber\\
			&= \dfrac{2j}{i}
		\end{align}
		
		So
		\begin{equation}
			E[F_{ij}] = \sum_{k\leq j}{E[D_{ik}]} = \dfrac{j(j+1)}{i}
		\end{equation}
		
		\item  
		
		Similar to \citet{janson2011total}, we will use indicator variables tracking whether a specific external branch is still present at each time interval. We will start by arbitrarily labeling the leaves and define $Z_{i,k}$ as the indicator random variable that leaf $i$ is still present when there are $k$ branches, when viewing it as a Tajima coalescent process, starting with $n$ leaves and successively merging two at a time. Then $F_{n-1,k-1}=\sum^{n}_{i=1}Z_{i,k}$ and for $k \leq l$,
		\begin{equation}
			\text{Cov}(F_{n-1,k-1},F_{n-1,l-1})=\sum^{n}_{i=1}\sum^{n}_{j=1}\text{Cov}(Z_{i,k},Z_{j,l}),
		\end{equation}
		where
		\begin{equation}
			\text{Cov}(Z_{i,k},Z_{j,l})=\text{E}(Z_{i,k}Z_{j,l})-\text{E}(Z_{i,k})\text{E}(Z_{j,l})
		\end{equation}
		\begin{equation}
			\text{E}(Z_{i,k}Z_{j,l})=\text{P}(Z_{i,k}=Z_{j,l}=1)
		\end{equation}
		is the probability that the external branch with label $i$ and the external branch with label $j$ remain external branches when there are $k$ and $l$ branches respectively. That is, the probability that leaves $i$ and $j$ do not coalesce when there are $n,n-1,\ldots,l$ branches and when one of them do not coalesce when there are $l-1,\ldots,k$. That is
		\begin{equation}
			\text{E}(Z_{i,k}Z_{j,l})=\text{P}(Z_{i,k}=Z_{j,l}=1)=\frac{\binom{n-2}{2}}{\binom{n}{2}}\frac{\binom{n-3}{2}}{\binom{n-1}{2}}\cdots\frac{\binom{l-1}{2}}{\binom{l+1}{2}}\frac{\binom{l-1}{2}}{\binom{l}{2}}\cdots\frac{\binom{k}{2}}{\binom{k+1}{2}}=\frac{(l-1)(l-2)k(k-1)}{n(n-1)^{2}(n-2)},
		\end{equation}
		and
		\begin{equation}
			\text{E}(Z_{i,k})=\text{P}(Z_{i,k}=1)=\frac{\binom{n-1}{2}}{\binom{n}{2}}\frac{\binom{n-2}{2}}{\binom{n-1}{2}}\cdots\frac{\binom{k}{2}}{\binom{k+1}{2}}=\frac{k(k-1)}{n(n-1)},
		\end{equation}
		\begin{equation}
			\text{Var}(Z_{i,k})=\frac{k(k-1)}{n(n-1)}\left( 1- \frac{k(k-1)}{n(n-1)}\right),
		\end{equation}
		
		\begin{align}
			\text{Cov}(Z_{i,k}Z_{j,l})& =\frac{(l-1)(l-2)k(k-1)}{n(n-1)^{2}(n-2)}-\frac{k(k-1)l(l-1)}{n^{2}(n-1)^{2}} \nonumber\\
			& =\frac{(l-1)(k-1)k}{n(n-1)^{2}}\left[\frac{l-2}{n-2}-\frac{l}{n}\right] \nonumber\\
			& =\frac{-2(n-l)(l-1)k(k-1)}{n^{2}(n-1)^{2}(n-2)} \nonumber\\
		\end{align}
		\begin{align}
			\text{Cov}(Z_{1,i},Z_{2,i})&=\text{E}(Z_{1,i}Z_{2,i})-\text{E}(Z_{1,i})\text{E}(Z_{2,i})\nonumber\\
			&=\frac{\binom{i}{2}\binom{i-1}{2}}{\binom{n}{2}\binom{n-1}{2}}-\frac{i^{2}(i-1)^2}{n^{2}(n-1)^{2}}=\frac{-2i(i-1)^{2}(n-i)}{n^{2}(n-1)^{2}(n-2)}\nonumber\\
		\end{align}
		and
		\begin{align}
			\text{Var}(F_{n-1,i-1})&=\sum^{n}_{j=1}\text{Var}(Z_{j,i})+n(n-1)\text{Cov}(Z_{1,i},Z_{2,i})\nonumber\\
			&=\frac{i(i-1)}{n-1}\left( 1- \frac{i(i-1)}{n(n-1)}\right)+\frac{-2i(i-1)^{2}(n-i)}{n(n-1)(n-2)}\nonumber\\
			&=\frac{i(i-1)}{n(n-1)^{2}(n-2)}\left[n(n-1)(n-2)-i(i-1)(n-2)-2(i-1)(n-i)(n-1) \right]\nonumber\\
			&=\frac{i(i-1)}{n-1}-\frac{i(i-1)^{2}}{n(n-1)^{2}(n-2)}\left[i(n-2)+2(n-i)(n-1) \right]\nonumber\\
			&=\frac{i(i-1)}{n-1}-\frac{i^{2}(i-1)^{2}}{n(n-1)^{2}}-\frac{2i(i-1)^{2}(n-i)}{n(n-1)(n-2)} \nonumber\\
			&=\frac{i^{2}(i-1)^{2}}{(n-1)^{2}(n-2)}+\frac{i(i-1)(n-2i)}{(n-1)(n-2)} 
		\end{align}
		
		So $$\operatorname{Var}[F_{ij}] = \dfrac{j^2(j+1)^2}{i^2(i-1)} + \dfrac{j(j+1)(i-2j-1)}{i(i-1)}$$
		
		\item We continue to show the results of covariance.
		First, when $k\leq l$
		
		\begin{align}
			\text{Cov}(F_{n-1,k-1},F_{n-1,l-1})&=\sum^{n}_{i=1}\sum^{n}_{j=1}\text{Cov}(Z_{i,k},Z_{j,l})=n(n-1)\text{E}(Z_{1,k}Z_{2,l})+n\text{E}(Z_{1,k})-n^{2}\text{E}(Z_{1,k})\text{E}(Z_{1,l})\nonumber\\
			&=\frac{(l-1)(l-2)k(k-1)}{(n-1)(n-2)}+\frac{k(k-1)}{n-1}-\frac{k(k-1)l(l-1)}{(n-1)^{2}}\nonumber\\
			&=\frac{k(k-1)[l(l+1)+n(n-2l-1)]}{(n-1)^{2}(n-2)}\nonumber\\
		\end{align}
		
		So when $i_1 = i_2, j_1 \leq j_2$, $$\operatorname{Cov}[F_{i_1 j_1}, F_{i_2 j_2}] = \frac{j_1(j_1+1)[j_2(j_2+2)+(i_1 + 1)(i_1-2j_2-2)]}{i_1^{2}(i_1 - 1)} $$

		Now, when comparing values of the $\mathbf{F}$ matrix at different rows, for example $F_{n-1,k-1}$ and $F_{m-1,k-1}$ for $k<n<m$, we then assume that $F_{m-1,k-1}=\sum^{m}_{i=1}Z_{i,k}$ denotes the number of external branches in a ranked tree shape with $m$ leaves when there are $k$ branches as before, however, when considering it together with $F_{n-1,k}$, $F_{n-1,k}$ denotes the number of external branches when there are $k$ branches in a ranked tree shape with $n$ leaves. That is, when there are $n$ branches, there are $\sum^{m}_{i=1}Z_{i,n}$ external branches with respect to the bigger tree but there are additional $n-\sum^{m}_{i=1}Z_{i,n}$ branches that are external with respect to the smaller tree. We then have
		\begin{align}
			\text{Cov}(F_{m-1,k-1},F_{n-1,k-1})&=\text{Cov}\left(F_{m-1,k-1},F_{m-1,k-1}+\sum^{2m-n}_{j=m+1}Z_{j,k}\right)\nonumber\\
			&=\text{Var}(F_{m-1,k-1})+m\sum^{2m-n}_{j=m+1}\text{Cov}\left(Z_{1,k},Z_{j,k}\right)
		\end{align}
		where
		\begin{align}
			\text{Cov}(Z_{1,k},Z_{m+1,k})&=E\left[ Z_{1,k}Z_{m+1,k}\right]-\text{E}[Z_{1,k}] \text{E}[Z_{m+1,k}]\nonumber\\
			&=\frac{\binom{m-1}{2}}{\binom{m}{2}}\frac{\binom{m-3}{2}}{\binom{m-1}{2}}\cdots\frac{\binom{k}{2}}{\binom{k+2}{2}}\frac{\binom{k-1}{2}}{\binom{k+1}{2}}-\frac{\binom{k}{2}}{\binom{m}{2}}\frac{\binom{k}{2}}{\binom{m-1}{2}}\nonumber\\
			&=\frac{\binom{k}{2}\binom{k-1}{2}}{\binom{m}{2}\binom{m-2}{2}}-\frac{\binom{k}{2}}{\binom{m}{2}}\frac{\binom{k}{2}}{\binom{m-1}{2}}\nonumber\\
			&=\frac{k(k-1)^{2}(k-2)}{m(m-1)(m-2)(m-3)}-\frac{k^{2}(k-1)^{2}}{m(m-1)^{2}(m-2)}
		\end{align}
		\begin{align}
			\text{Cov}(Z_{1,k},Z_{m+2,k})&=\frac{\binom{m-1}{2}}{\binom{m}{2}}\frac{\binom{m-2}{2}}{\binom{m-1}{2}}\frac{\binom{m-4}{2}}{\binom{m-2}{2}}\cdots\frac{\binom{k}{2}}{\binom{k+2}{2}}\frac{\binom{k-1}{2}}{\binom{k+1}{2}}-\frac{\binom{k}{2}}{\binom{m}{2}}\frac{\binom{k}{2}}{\binom{m-2}{2}}\nonumber\\
			&=\frac{\binom{k}{2}\binom{k-1}{2}}{\binom{m}{2}\binom{m-3}{2}}-\frac{\binom{k}{2}}{\binom{m}{2}}\frac{\binom{k}{2}}{\binom{m-2}{2}}\nonumber\\
			&=\frac{k(k-1)^{2}(k-2)}{m(m-1)(m-3)(m-4)}-\frac{k^{2}(k-1)^{2}}{m(m-1)(m-2)(m-3)}
		\end{align}
		and
		\begin{align}
			\sum^{2m-n}_{j=m+1}\text{Cov}\left(Z_{1,k},Z_{j,k}\right)&=\sum^{m-n+1}_{l=2}\left\lbrace \frac{k(k-1)^{2}(k-2)}{m(m-1)(m-l)(m-l-1)}-\frac{k^{2}(k-1)^{2}}{m(m-1)(m+1-l)(m-l)} \right\rbrace
			\nonumber\\
			&=\frac{k(k-1)^{2}(k-2)(m-n)}{m(m-1)(m-2)(n-2)}-\frac{k^{2}(k-1)^{2}(m-n)}{m(m-1)^{2}(n-1)}\nonumber\\
			&=\frac{k(k-1)^{2}(m-n)}{m(m-1)}\left[ \frac{k-2}{(m-2)(n-2)}-\frac{k}{(m-1)(n-1)}\right]
		\end{align}

		Therefore,
		
		\begin{align}
			\text{Cov}(F_{m-1,k-1},F_{n-1,k-1})&=\frac{k^{2}(k-1)^{2}}{(m-1)^{2}(m-2)}+\frac{k(k-1)(m-2k)}{(m-1)(m-2)}\nonumber\\
			& \hspace{1cm} +\frac{k(k-1)^{2}(m-n)}{(m-1)}\left[ \frac{k-2}{(m-2)(n-2)}-\frac{k}{(m-1)(n-1)}\right] \nonumber \\
			& = \frac{k (k-1) (k-n) (k-n+1)}{(m-1) (n-1) (n-2)}
		\end{align}
		
		So when $i_1 > i_2, j_1 = j_2$, $$\operatorname{Cov}[F_{i_1 j_1}, F_{i_2 j_2}] = \frac{j_2 (j_2 + 1) (j_2-i_2) (j_2-i_2+1)}{i_1 i_2 (i_2-1)} $$

		For $i<j$, $i<n$ and $j<m$
		\begin{align}
			\text{Cov}(F_{n-1,i-1},F_{m-1,j-1})&=\text{Cov}(F_{m-1,i,-1}+\sum^{2m-n}_{l=m+1}Z_{l,i},F_{m-1,j-1})\nonumber\\
			&=\text{Cov}(F_{m-1,i-1},F_{m-1,j-1})+\text{Cov}(F_{m-1,j-1},\sum^{2m-n}_{l=m+1}Z_{l,i})
		\end{align}
		
		where
		
		\begin{align}
			\text{Cov}(F_{m-1,j-1},\sum^{2m-n}_{l=m+1} Z_{l,i}) &= m(m-n)\text{Cov}(Z_{1,j},Z_{m+1,i})\nonumber\\
			&=m(m-n)\left[ \text{E}[Z_{1,j} Z_{m+1,i} ]-\frac{j(j-1)}{m(m-1)}\frac{i(i-1)}{n(n-1)}\right]\nonumber\\
			&=m(m-n)\left[ \frac{(j-1)(j-2)}{m(m-1)}\frac{i(i-1)}{(n-1)(n-2)}-\frac{j(j-1)}{m(m-1)}\frac{i(i-1)}{n(n-1)}\right]\nonumber\\
			&=\frac{(m-n)(j-1)i(i-1)}{(m-1)(n-1)}\left[ \frac{j-2}{n-2}-\frac{j}{n}\right]\nonumber\\
		\end{align}
		Then, for $i<j$, $i<n$, $j<m$, and $m>n$
		
		\begin{align}
			\text{Cov}(F_{n-1,i-1},F_{m-1,j-1})&=\frac{i(i-1)[j(j+1)+m(m-2j-1)]}{(m-1)^{2}(m-2)}+\frac{(m-n)(j-1)i(i-1)}{(m-1)(n-1)}\left[ \frac{j-2}{n-2}-\frac{j}{n}\right]\nonumber\\
		\end{align}
		
		So when $i_1 > i_2, j_1 > j_2$, $$\operatorname{Cov}[F_{i_1 j_1}, F_{i_2 j_2}] =\frac{j_2(j_2+1)[(j_1+1)(j_1+2)+(i_1+1)(i_1-2j_1-2)]}{i_1^{2}(i_1-1)}+\frac{(i_1-i_2)j_1j_2(j_2+1)}{i_1i_2}\left[ \frac{j_1-1}{i_2-1}-\frac{j_1+1}{i_2+1}\right] $$
		
		Now, for $j<i$, $i<n$, $j<m$, and $m>n$
		\begin{align}
			\text{Cov}(F_{n-1,i-1},F_{m-1,j-1})&=\text{Cov}(F_{m-1,i-1}+\sum^{2m-n}_{l=m+1}Z_{l,i},F_{m-1,j-1})\nonumber\\
			&=\text{Cov}(F_{m-1,i-1},F_{m-1,j-1})+\text{Cov}(F_{m-1,j-1},\sum^{2m-n}_{l=m+1}Z_{l,i})\nonumber\\
		\end{align}
		
		where
		
		\begin{align}
			\text{Cov}(F_{m-1,j-1},\sum^{2m-n}_{l=m+1} Z_{l,i}) &= m(m-n)\text{Cov}(Z_{1,j},Z_{m+1,i})\nonumber\\
			&=m(m-n)\left[ \text{E}[Z_{1,j} Z_{m+1,i} ]-\frac{j(j-1)}{m(m-1)}\frac{i(i-1)}{n(n-1)}\right]\nonumber\\
			&=m(m-n)\left[ \frac{(i-1)(i-2)}{m(m-1)}\frac{j(j-1)}{(n-1)(n-2)}-\frac{j(j-1)}{m(m-1)}\frac{i(i-1)}{n(n-1)}\right]\nonumber\\
			&=\frac{(m-n)(j-1)j(i-1)}{(m-1)(n-1)}\left[ \frac{i-2}{n-2}-\frac{i}{n}\right]\nonumber\\
		\end{align}
		
		Then, for $j<i$, $i<n$, $j<m$, and $m>n$:
		\begin{align}
			\text{Cov}(F_{n-1,i-1},F_{m-1,j-1})&=\frac{j(j-1)[i(i+1)+m(m-2i-1)]}{(m-1)^{2}(m-2)}+\frac{(m-n)(j-1)j(i-1)}{(m-1)(n-1)}\left[ \frac{i-2}{n-2}-\frac{i}{n}\right]\nonumber\\
		\end{align}
		
		So when $i_1 > i_2, j_1 < j_2$, $$\operatorname{Cov}[F_{i_1 j_1}, F_{i_2 j_2}] = \frac{j_1(j_1+1)[(j_2+1)(j_2+2)+(i_1+1)(i_1-2j_2-2)]}{i_1^{2}(i_1-1)}+\frac{(i_1-i_2)j_1(j_1+1)j_2}{i_1i_2}\left[ \frac{j_2-1}{i_2-1}-\frac{j_2+1}{i_2+1}\right]$$

	\end{enumerate}
\end{proof}




\section{Temperature schedules for Simulated Annealing} \label{app:annealing_params}

We now consider the choice of temperature schedule in the Simulated Annealing (SA) algorithm of Section \ref{sec:sa}.
in the main document.  The parameters to be chosen are whether we use a logartihmic, linear, or exponential cooling regime, as well as the initial temperature $R_0$, and cooling rate $\alpha$.  Theoretical convergence guarantees exist for the logarithmic cooling schedule $R_k = R_0(1 + \alpha \log(1 + k))^{-1}$
with sufficiently high initial temperature and appropriately chosen $\alpha$ (see Chapter 3 of \citet{aarts1988simulated}).  
Other options include the linear schedule $R_k = R_0(1 + \alpha k)^{-1}$, or the exponential cooling schedule $R_k = R_0 \alpha^k$.  
In practice, the logartithmic schedule is very inefficient and better choices may be made for each different problem.

In our simulations, we observed that the exponential cooling schedule with $\alpha = .9995$  works well for trees with upto $n=250$ tips.  We note the logarithmic schedule is prohibitively slow, and we do not gain much by using the linear schedule over the exponential schedule.

\begin{figure}
	\centering
	\includegraphics[width=0.5\textwidth]{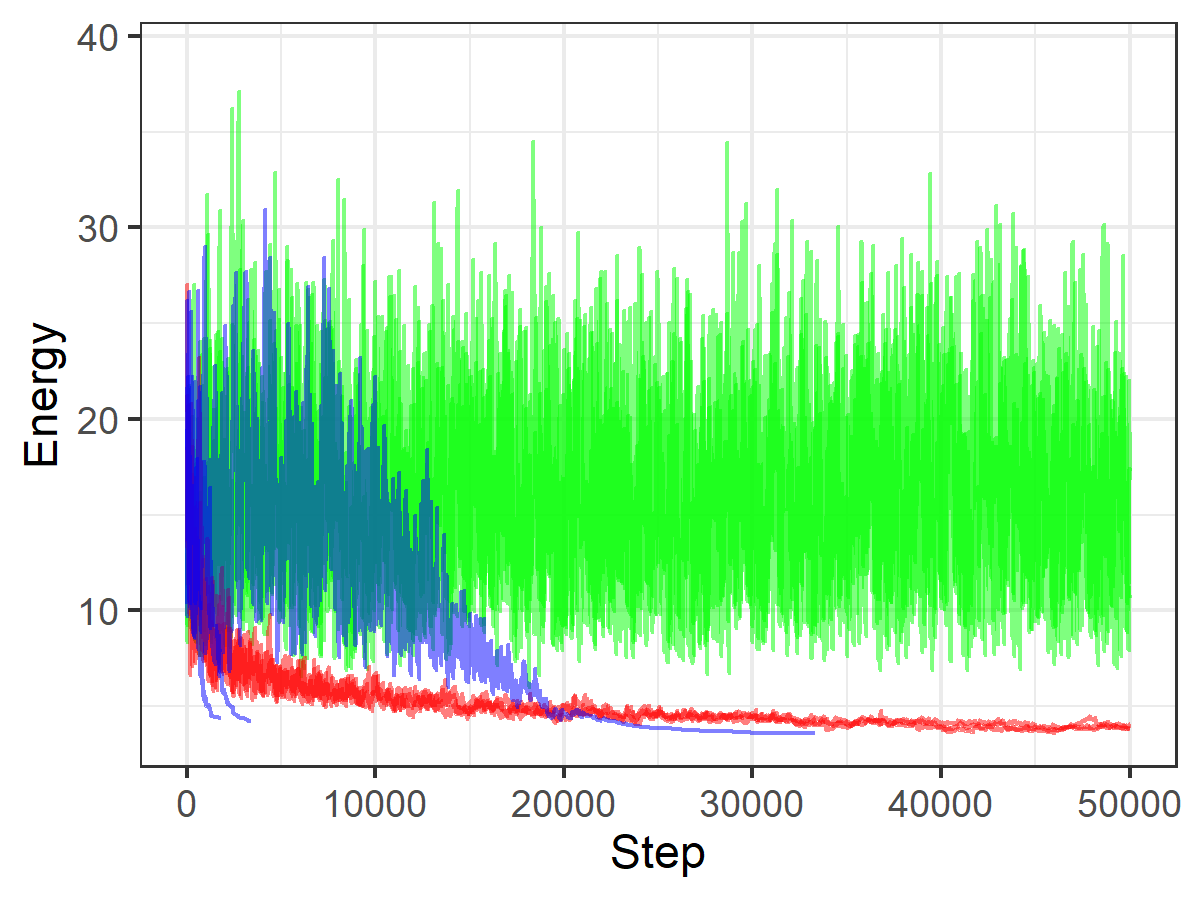}
	\caption{Example energy runs a simulated annealing chain, with trees on $n = 100$ tips.  Three runs each for the logarithmic (green), linear (red), and exponential (blue) schedules.  We observe that the exponential schedule with $\alpha = .9995$ has the best performance.}
	\label{fig:energy}
\end{figure}


\section{SARS-CoV-2 Posterior Inference and Sample Information}

We accessed a global alignment of SARS-CoV-2 publicly available molecular sequences obtained from infected human individuals on November 4, 2020 in the  GISAID EpiCov database \citep{shu2017gisaid}. We analyzed a subset of available sequences sampled in the states of California, Florida, Texas and Washington in the USA for the period of February 2020 to September 2020. We analyzed 19 random samples of 100 sequences in California and 3 random samples of 100 sequences in the rest of the states. We independently sampled genealogies from the posterior distribution for each random sample independently with BEAST \citep{Suchard2018}. In BEAST, we placed the HKY mutation prior on the substitution process with fixed global mutation rate of $8\times 10^{-4}$ substitutions per site per year as it has been done in other studies \citep{Nadeaue2012008118}, and we placed the Skyride process prior on the coalescent effective population size \citep{minin2008smooth}. We ran each chain for $50$ million iterations and thinned every $5000$ iterations.

We acknowledge the following submitting and originating laboratories that made the sequences used in this manuscript available in GISAID: San Diego County Public Health Laboratory, Quest Diagnostics, UCLA Pathology Clinical Microbiology Lab, Cedars-Sinai Medical Center, Department of Pathology \& Laboratory Medicine, Molecular Pathology Laboratory, San Joaquin County Public Health Lab, Orange County Public Health Laboratory, UCSF Clinical Microbiology Laboratory, UC San Diego Center for Advanced Laboratory Medicine, County of Santa Clara Public Health, Scripps Medical Laboratory, California Department of Public Health, Ventura County Public Health Lab, Rady's Childrens Hospital, Stanford clinical virology lab, Santa Clara County Public Health Department, County of Santa Clara Public Health Department, Humboldt County Public Health Laboratory, Department of Immunology, Contra Costa Public Health Lab, San Francisco Public Health Laboratory, National Institute for Communicable Disease Control and Prevention (ICDC) Chinese Center for Disease Control and Prevention (China CDC), Innovative Genomics Institute, UC Berkeley, San Bernardino County Public Health Lab, Alameda County Public Health Lab, California Department of Health, County Of San Luis Obispo Public Health Laboratory, University of California, Davis, Andersen Lab, Microbial Diseases Laboratory, University of Wisconsin-Madison AIDS Vaccine Research Laboratories, OHSU Lab Services Molecular Microbiology Lab, San Luis Obispo Public Health Department, Tuolumne County Public Health, FL Bureau of Health Laboratories Tampa, Florida Bureau of Public Health Laboratories, Florida Department of Health, Public Health, United States Air Force School of Aerospace Medicine, University of Florida, FL Bureau of Public Health Laboratories-Miami, University of Miami Immunology and Histocompatibility Laboratory, FL Bureau of Public Health Laboratories, Laboratory of Dr. John Lednicky, Environmental and Global Health, University of Florida - Gainesville, FL Bureau of Public Health Laboratories-Tampa, Mayo Clinic Laboratories, University of Florida, FL Bur. of Public Health Laboratories-Jacksonville, Emerging Pathogens Institute, University of Florida, Environmental and Global Health, Texas DSHS Lab Services, Houston Methodist Hospital, Texas Department of State Health Services, LSUHS Emerging Viral Threat Laboratory, Texas Department of State Health Services (TXDSHS), Baylor College of Medicine, City of El Paso Department of Public Health Laboratory, Washington State Department of Health, University of Washington Virology Lab, WA State Department of Health, University of Washington, Laboratory Medicine, Laboratory Medicine, University of Washington, Washington State Public Health Lab, Harborview Medical Center, Andersen lab at Scripps Research, Kruglyak Lab, Cedars-Sinai Medical Center, Molecular Pathology Laboratory of Department of Pathology \& Laboratory Medicine and Genomic Core, Chan-Zuckerberg Biohub, Chiu Laboratory, University of California, San Francisco, Q Squared Solutions - QRTP facility, Pathogen Discovery, Respiratory Viruses Branch, Division of Viral Diseases, Centers for Disease Control and Prevention, Oregon SARS-CoV-2 Genome Sequencing Center, Ginkgo Bioworks Clinical Laboratory, University of Florida, Microbial Genome Sequencing Center.

\label{app:covid_data}

\end{document}